\documentclass[submission]{eptcs}
\def\event{Gandalf 2019}

\usepackage{tikz}
\usetikzlibrary{automata}
\usetikzlibrary {positioning}
\usepackage{textcomp}

\usepackage{multirow}
\usepackage{multicol}

\usepackage[linesnumbered,lined,commentsnumbered,ruled,scleft]{algorithm2e}

\usepackage{thm-restate}
\usepackage{nico}
\usepackage{macros}

\usepackage[T1]{fontenc}
\usepackage{microtype}

\begin{document}
\title{Reachability Games with Relaxed Energy Constraints}
\def\titlerunning{Reachability Games with Relaxed Energy Constraints}
\author{Lo\"ic H\'elou\"et\hskip2cm  Nicolas Markey\institute{Inria \& CNRS \& Univ. Rennes, France}
  \and
  Ritam Raha\institute{Chennai Mathematical Institute, India}}
\def\authorrunning{Lo\"ic H\'elou\"et, Nicolas Markey, and Ritam Raha}

\maketitle

\begin{abstract}
We study games with reachability objectives under energy constraints.
We~first prove that under strict energy constraints (either only
lower-bound constraint or interval constraint), those games are
\LOGSPACE-equivalent to energy games with the same energy constraints but without reachability objective (i.e., for infinite
runs). We~then consider two kinds of
relaxations of the upper-bound constraints (while keeping the
lower-bound constraint strict): in the first one, called \emph{weak
upper bound}, the upper bound is \emph{absorbing}, in the sense that
it allows receiving more energy when the upper bound is already
reached, but the extra energy will not be stored; in the second one,
we~allow for \emph{temporary violations} of the upper bound, imposing
limits on the number or on the amount of violations.

We prove that when considering weak upper bound, reachability
objectives require memory, but can still be solved in polynomial-time
for one-player arenas; we~prove that they are in \co\NP in the two-player
setting. Allowing for bounded violations makes the
problem \PSPACE-complete for one-player arenas and \EXPTIME-complete
for two players.

\end{abstract}

\vspace{-\medskipamount}
\section{Introduction}

\paragraph{Games on weighted graphs.}
Weighted games are a common way to formally address questions
related to consumption, production and storage of resources: the~arenas
of such game are two-player turn-based games in which transitions
carry positive or negative integers, representing the accumulation or
consumption of resource.  Various objectives have been considered for
such arenas, such as optimizing the total or average amount of
resources that have been collected along the play, or maintaining the
total amount within given bounds. The~latter kind of objectives,
usually referred to as \emph{energy
objectives}~\cite{CdAHS03,BouyerFLMS08}, has been widely studied in
the untimed
setting~\cite{ChatterjeeD12,ChatterjeeDHR10,DDGRT10,FahrenbergJLS11,JLR13,
JLS15,VCDHRR15,BMRLL15,BHMRZ17,DM18}, and to a lesser extent in the
timed setting~\cite{BFLM10,qest2012-BLM}.
As their name indicates, energy objectives can be used to model the
evolution of the available energy in an autonomous system: besides
achieving its tasks, the~system has to take care of recharging
batteries regularly enough so as to never run out of power.  Energy
objectives were also used to model moulding machines: such machines
inject molten plastic into a mould, using pressure obtained by storing
liquid in a tank~\cite{CJLRR09}; the~level of liquid has to be
controlled in such a way that enough pressure is always available, but
excessive pressure in the tank would reduce the service life of the valve.


Energy games impose strict constraints on the total amount of energy
at all stages of the play. Two kinds of constraints have been mainly
considered in the literature: lower-bound constraints (a.k.a. \Lenergy
constraints) impose a strict lower bound (usually~zero), but impose no
upper bound; on~the other hand, lower- and upper-bound constraints
(a.k.a. \LUenergy constraints) require that the energy level always
remains within a bounded interval~$[L;U]$. Finding strategies that
realize \Lenergy objectives along infinite runs is in \PTIME in the
one-player setting, and in $\NP\cap\co\NP$ for two players;
for \LUenergy objectives, it~is respectively \PSPACE-complete
and \EXPTIME-complete~\cite{BouyerFLMS08}. Some works have also considered
the existence of an initial energy level for which a winning strategy
exist~\cite{ChatterjeeDHR10}.
%

In this paper, we focus on weighted games combining energy objectives
together with reachability objectives. Our~first result is the
(expected) proof that \Lenergy games with or without reachability
objectives are interreducible; the same holds for \LUenergy games.
We~then focus on relaxations of the energy constraints, in two
different directions. In both cases, the lower bound remains
unchanged, as it corresponds to running out of energy, which we always
want to avoid. We~thus only relax the upper-bound constraint.
The~first direction concerns \emph{weak upper bounds}, already
introduced in~\cite{BouyerFLMS08}: in~that setting, hitting the upper
bound is allowed, but there is no overload: trying to exceed the upper
bound will simply maintain the energy level at this maximal
level. Yet, a~strict lower bound is still
imposed. Following~\cite{BouyerFLMS08}, we~name these
objectives \LWenergy objectives. They~could be used as a (simplified)
model for batteries.  When considered alone, \LWenergy objectives are
not much different from~\Lenergy objectives, in the sense that the aim
is to find a reachable \emph{positive loop}. \LWenergy games (with no
other objectives besides maintaining the energy above~$L$) are
in \PTIME for one-player games, and in $\NP\cap\co\NP$ for two
players~\cite{BouyerFLMS08}. When combining \LWenergy and reachability
objectives, we~prove in this paper that the~situation changes:
different loops may have different effects on the energy level, and we
have to keep track of the final energy level reached when iterating
those loops.

We~introduce and study a second way of relaxing upper bounds, which we
call \emph{soft upper bound}: it~consists in allowing a limited number
(or~amount) of violations of the soft upper bound (possibly within an
additional strong upper-bound): when modeling a pressure tank,
the~lower-bound constraint is strict (pressure should always be
available) but the upper bound is soft (excessive pressure may be
temporarily allowed if needed). We~consider different kinds of restrictions
(on~the number or amount of violations), and prove that
deciding whether \Pl1 has a strategy to keep violations below a given bound is
\PSPACE-complete for one-player arenas,
and \EXPTIME-complete for two-player ones.
We~also provide algorithms to optimize violations of the soft upper bound
under a given strict upper bound.

\paragraph*{Related work.}
Quantitative games have been the focus of numerous research articles
since the~1970s, with various kinds of objectives, such as ultimately
optimizing the total payoff, mean-payoff~\cite{EM79,ZwickP95},
or discounted sum~\cite{ZwickP95,Andersson06}.
Energy objectives, which are a kind of safety objectives on the total
payoff, were introduced in~\cite{CdAHS03} and rediscovered
in~\cite{BouyerFLMS08}.
Several works have extended those works by
combining quantitative conditions together, e.g. multi-dimensional
energy conditions~\cite{FahrenbergJLS11,JLS15} or
conjunctions of energy- and mean-payoff
objectives~\cite{ChatterjeeDHR10}. Combinations with qualitative
objectives (e.g. reachability~\cite{Chatterjee0H17} or parity
objectives~\cite{ChatterjeeD12,ChatterjeeRR14,DM18}) were also considered.
Similar objectives have been considered in slightly different
settings e.g. Vector Addition Systems with States~\cite{Reichert16} and
one-counter machines~\cite{GHOW10,Hun15}.

\vspace{-\medskipamount}
\section{Preliminaries}
\label{section_prelim}

\begin{definition}
A two-player arena is a $3$-tuple $G=\tuple{Q_1,Q_2,E}$ where
$Q=Q_1\uplus Q_2$ is a set of states, $E\subseteq Q \times \bbZ\times
Q$ is a set of weighted edges.
For~$q\in Q$, we~let $qE=\{(q,w,q')\in E
\mid w\in\bbZ,\ q'\in Q\}$, which we assume is non-empty for any~$q\in Q$.
A~one-player arena is a two-player arena where $Q_2=\emptyset$.
\end{definition} 

Consider a state~$q_0\in Q$.  A \emph{finite path} in an arena~$G$
from an initial state~$q_0$ is an finite sequence of edges $\pi =
(e_i)_{0\leq i< n}$ such that for every $0\leq i<n$, writing
$e_i=(q_i,w_i,q'_i)$, it~holds $q'_i=q_{i+1}$.  Fix a path $\pi =
(e_i)_{0\leq i< n}$.  Using the notations above, we~write $\size \pi$
for the size~$n$ of~$\pi$, $\hat\pi_i$~for the $i$-th state~$q_i$ of~$\pi$
(with the convention that $q_n=q'_{n-1}$), and $\first(\pi)=\hat\pi_0$ for its
first state and $\last(\pi)=\hat\pi_n$ for its last state.
The~empty path is a special finite path from~$q_0$; its~length
is~zero, and $q_0$ is both its first and last state.
Given two finite paths~$\pi=(e_i)_{0\leq i<n}$
and~$\pi'=(e'_j)_{0\leq j< n'}$ such that
$\last(\pi)=\first(\pi')$, the concatenation~$\pi\cdot\pi'$
is the finite path $(f_k)_{0\leq k<n+n'}$ such that $f_k=e_k$ if~$0\leq
k<n$ and $f_k=e'_{k-n}$ if $n\leq k<n+n'$.
%
For~$0\leq k\leq n$, the $k$-th prefix of~$\pi$ is the finite path
$\pi_{<k}=(e_i)_{0\leq i< k}$.  We~write $\FPath(G,q_0)$ for the set
of finite paths in~$G$ issued from~$q_0$ (we~may omit to mention~$G$
in this notation when it is clear from the context).  Infinite paths
are defined analogously; we~write $\Path(G,q_0)$ for the set of
infinite paths from~$q_0$.

A~\emph{strategy} for \Pl1 (resp.~\Pl2) from~$q_0$ is a function~$\sigma\colon
\FPath(q_0)\to E$ associating with any finite path~$\pi$ with
$\last(\pi)\in Q_1$ (resp.~$\last(\pi)\in Q_2$) an edge originating
from~$\last(\pi)$. 
A~strategy is said~\emph{memoryless} when $\sigma(\pi)=\sigma(\pi')$
whenever $\last(\pi)=\last(\pi')$.

A~finite path $\pi = (e_i)_{0\leq i<n}$ \emph{conforms} to a
strategy~$\sigma$ of \Pl1 (resp.~of~\Pl2) from~$q_0$ if
$\first(\pi)=q_0$ and for every $0\leq k<n$, either
$e_{k}=\sigma(\pi_{<k})$, or $\last(\pi_{<k})\in Q_2$
(resp.~$\last(\pi_{<k})\in Q_1$).  This is extended to infinite paths
in the natural way. Given a strategy~$\sigma$ of~\Pl1 (resp.~of~\Pl2)
from~$q_0$, the~outcomes of~$\sigma$ is the set of all (finite or infinite)
paths~$\pi$
issued from~$q_0$ that conform to~$\sigma$.



A game is a triple~$(G,\qinit,\obj)$ where $G$ is a two-player arena,
$\qinit$ is an initial state in~$Q$, and $\obj\subseteq \Path(G,\qinit)$
is a set of infinite paths, also called \emph{objective} (for~\Pl1).
A~strategy for \Pl1 from~$\qinit$ is winning in~$(G,\qinit,\obj)$ if its
infinite outcomes all belong to~$\obj$.

Given a set~$R\subseteq Q$ of states, the reachability objective
defined by~$R$ is the set of all paths containing some state in~$R$,
while the safety objective defined by~$R$ is the set of all infinite
paths never visiting any state in~$R$. In~this paper, we also focus on
\emph{energy objectives}~\cite{CdAHS03,BouyerFLMS08}, which we now define.
%

\begin{definition}
  \looseness=-1
  Fix a finite-state arena~$G=\tuple{Q_1,Q_2,E}$. Let
  $L\in\bbZ$. The~\Lenergy arena associated with~$G$ is the infinite
  arena $\exGL=\tuple{C_1,C_2,T}$ where $C_1=\{\qerr\}\cup
  Q_1\times[L;+\infty)$ and $C_2=Q_2\times[L;+\infty)$ are sets
      of~\emph{configurations}, and $T\subseteq C_1\times\bbZ\times
      C_2$ is such that
  \begin{itemize}
  \item for any~$(q,l)$ and~$(q',l')$ in~$Q\times[L;+\infty)$ and
    any~$w\in \bbZ$, we~have $((q,l),w,(q',l'))\in T$ if, and only~if,
    $(q,w,q')\in E$ and $l'=l+w\geq L$; we~also impose a loop
    $(\qerr,0,\qerr)\in T$.
  \item for any~$(q,l)\in Q\times[L;+\infty)$, we~have
    $((q,l),w,\qerr)\in T$ if, and only~if, there exists a
    transition~$(q,w,q')\in E$ such that $l+w<L$
  \end{itemize}
  Similarly, given $L\in\bbZ$ and $U\in\bbZ$, the \LUenergy arena
  associated with~$G$ is the finite-state arena
  $\exGLU=\tuple{C_1,C_2,T}$ where $C_1=\{\qerr\}\cup Q_1\times[L;U]$ and
    $C_2=Q_2\times[L;U]$, and $T\subseteq C_1\times\bbZ\times
  C_2$ is such that
  \begin{itemize}
  \item for any~$(q,l)$ and~$(q',l')$ in~$Q\times[L;U]$ and any~$w\in
    \bbZ$, we~have $((q,l),w,(q',l'))\in T$ if, and only~if,
    $(q,w,q')\in E$ and $l'=l+w\in [L;U]$; we~also impose a loop $(\qerr,0,\qerr)\in T$.
  \item for any~$(q,l)\in Q\times[L;U]$, we~have
    $((q,l),w,\qerr)\in T$ if, and only~if, there is a
    transition~$(q,w,q')\in E$ such that $l+w<L$ or $l+w>U$;
  \end{itemize}
  Finally, given $L\in\bbZ$ and $W\in\bbZ$, the \LWenergy arena
  associated with~$G$ is the finite-state arena
  $\exGLW=\tuple{C_1,C_2,T}$ where $C_1=\{\qerr\}\cup Q_1\times[L;W]$ and
    $C_2=Q_2\times[L;W]$, and $T\subseteq C_1\times\bbZ\times
  C_2$ is such that
  \begin{itemize}
  \item for any~$(q,l)$ and~$(q',l')$ in~$Q\times[L;W]$ and
    any~$w\in \bbZ$, we~have $((q,l),w,(q',l'))\in T$ if, and only~if,
    $(q,w,q')\in E$ and $l'=\min(W,l+w)\geq L$; we~also impose a loop $(\qerr,0,\qerr)\in T$.
  \item for any~$(q,l)\in Q\times[L;W]$, we~have
    $((q,l),w,\qerr)\in T$ if, and only~if, there is a
    transition~$(q,w,q')\in E$ such that $l+w<L$;
  
  \end{itemize}
%

  An \Lrun (resp. \LUrun, \LWrun)~$\rho$ in~$G$ from~$q$ with initial
  energy level~$l$ is a path in~$\exGL$ (resp~$\exGLU$, $\exGLW$)
  from~$(q,l)$ never visiting~$\qerr$. With~such a run~$\rho=(t_i)_i$
  in~$G$, writing $t_i=((q_i,l_i),w_i,(q'_i,l'_i))$, we~associate the
  path~$\pi=(e_i)_i$ such that $e_i=(q_i,w_i,q'_i)$.  We~define
  $\hat\rho_i=(q_i,l_i)$, corresponding to
  the $i$-th configuration along~$\rho$,
  and $\tilde\rho_i=l_i$, which we name the energy level in that configuration.

  Similarly, a path~$\pi$ is said \Lfeasible (resp. \LUfeasible,
  \LWfeasible) from~initial energy level~$L$ if there exists an \Lrun
  (resp. \LUrun, \LWrun) from~$(\first(\pi),L)$ whose associated path
  is~$\pi$. Notice that if such a run exists, it~is unique (because
  paths are defined as sequences of transitions).
\end{definition}

The \Lenergy (resp. \LUenergy, \LWenergy) objective is the set of
infinite paths that are \Lfeasible (resp. \LUfeasible, \LWfeasible)
(from initial energy level~$L$). Similarly, given a target
set~$R\subseteq Q$, the \Lenergy- (resp. \LUenergy-,
\LWenergy-) reachability objective is the set of \Lfeasible
(resp. \LUfeasible, \LWfeasible) paths visiting~$R$.

\begin{remark}
  Taking $L$ 
  as the initial energy level results in no loss of generality,
  since any energy level can be obtained by adding an initial
  transition from $(q_0,L)$.
\end{remark}

In many cases, strong upper bounds are too strict, as many system do
not break as soon as their maximal energy level is reached.  Imposing
a~\emph{weak} upper bound is a way to relax these constraints.
We~introduce another way to relax energy constraints, by allowing for
(limited) violations of the upper bound: given two strict bounds~$L$
and~$U$ in~$\bbZ$, a~soft upper bound~$S\in \bbZ$ with $L\leq S\leq
U$, and an~\LUrun $\rho$, the~set of
violations along~$\rho$ is the set $\VV(\rho) = \{i \in
[0;\size\rho]\mid \tilde\rho_i>S\}$ of positions along~$\rho$ where
the energy level exceeds the soft upper bound~$S$.
There are many ways to quantify violations along
a run. We~consider three of them in this paper, namely the total
number of violations, the maximal number of consecutive violations,
and the sum of the violations. We~thus define the following three
quantities:
%
$\nbV(\rho) = |\VV(\rho)|$,   $\consnbV(\rho)  = \max\{i-j+1 \mid \forall k\in[i,j].\ k\in \VV(\rho)\}$, and $\hV(\rho) = \sum_{i \in \VV(\rho)} (\tilde\rho_i-S)$.

Figure~\ref{fig_criteria} shows the evolution of $\nbV$ along a
winning run in an \LVenergynb game.
One can notice that
with a strong upper bound of~$3$, state~$q_t$ would not be reachable.
On~the other hand, if the strong upper bound is set to~$6$, then then
there exists a run from~$q_0$ to~$q_t$, but that requires $3$
violations of soft upper bound $S=3$ (and the total amount of
violations is~$6$).
\begin{figure}[htbp]
\begin{center}
\begin{tikzpicture}[node distance=1.8cm,
                    semithick]
\begin{scope}
\node[rond5,rouge] (q0) at (0,0) {$q_0$};
\node[rond5,jaune] (q1) [right of =q0] {$q_1$};
\node[rond5,jaune] (q2) [right of =q1] {$q_2$};
\node[rond5,jaune] (q3) [right of =q2] {$q_3$};
\node[rond5,vert] (qT) [below of=q3] {$q_t$};

\node at (1.5,-1.5) {$L=0$, $S=3$, $U=6$};

\path [-latex] (q0) edge node [yshift=0.5cm] {$+2$} (q1) {};
\path [-latex] (q1) edge [loop above] node  {$+1$} (q1){}; 
\path [-latex] (q1) edge node [yshift=0.5cm]{$-3$} (q2) {};
\path [-latex] (q2) edge[bend right] node [yshift=-0.5cm] {$+1$} (q3) {};
\path [-latex] (q3) edge[bend right] node [yshift=0.5cm] {$+2$} (q2) {};
\path [-latex] (q3) edge[bend left] node [xshift=0.5cm]{$-5$} (qT) {};
\path [-latex] (q3) edge [loop above] node {$-1$} (q3){};

\path [-latex] (qT) edge [loop right] node {0} (qT){}; 
\end{scope}
\begin{scope}[xshift=8cm,yshift=-1.5cm,xscale=1.2,inner sep=0pt]
\draw[->] (-0.2,0) -- (5.5,0) node[below right] {$\rho$}; 
\draw[->] (0,-0.2) -- (0,3.4) node[left] {energy};

\foreach  \x in {0.5,1,1.5,2,2.5,3,3.5,4,4.5,5}{
\draw[] (\x,-0.1) -- (\x,0.1){};
}

\foreach  \y/\t in {0/0,0.5/1,1/2,1.5/3,2/4,2.5/5,3/6}{
\draw[] (-0.1,\y) node[left] {$\t$} -- (0.1,\y){};
}

\draw[dashed] (0,1.5) -- (5.5,1.5){};
\draw[dashed] (0,3) -- (5.5,3){};
\node (lu) at (5.8,1.5) {$S$};
\node (lh) at (5.8,3) {$U$};

\node[] (l0) at (0,-0.5) {$q_0$};
\node[red] (p0) at (0,0) {$\bullet$};

\node[] (l1) at (0.5,-0.5) {$q_1$};
\node[yellow!70!red] (p1) at (0.5,1) {$\bullet$};
\draw[] (p0)--(p1){};

\node[] (l2) at (1,-0.5) {$q_1$};
\node[yellow!70!red] (p2) at (1,1.5) {$\bullet$};
\draw (p1)--(p2){};

\node[] (l3) at (1.5,-0.5) {$q_2$};
\node[yellow!70!red] (p3) at (1.5,0) {$\bullet$};
\draw (p2)--(p3){};

\node[] (l4) at (2,-0.5) {$q_3$};
\node[yellow!70!red] (p4) at (2,0.5) {$\bullet$};
\draw (p3)--(p4){};

\node[] (l5) at (2.5,-0.5) {$q_2$};
\node[yellow!70!red] (p5) at (2.5,1.5) {$\bullet$};
\draw (p4)--(p5){};

\node[] (l6) at (3,-0.5) {$q_3$};
\node[yellow!70!red] (p6) at (3,2) {$\bullet$};
\draw (p5)--(p6){};

\node[] (l7) at (3.5,-0.5) {$q_3$};
\node[yellow!70!red] (p7) at (3.5,1.5) {$\bullet$};
\draw (p6)--(p7){};

\node[] (l8) at (4,-0.5) {$q_2$};
\node[yellow!70!red] (p8) at (4,2.5) {$\bullet$};
\draw (p7)--(p8){};

\node[] (l9) at (4.5,-0.5) {$q_3$};
\node[yellow!70!red] (p9) at (4.5,3) {$\bullet$};
\draw (p8)--(p9){};

\node[] (l10) at (5,-0.5) {$q_t$};
\node[green!70!black] (p10) at (5,0.5) {$\bullet$};
\draw (p9)--(p10){};



\draw[line width=2pt,red] (p6) -- +(0,-.5);
\draw[line width=2pt,red] (p8) -- +(0,-1);
\draw[line width=2pt,red] (p9) -- +(0,-1.5);

\path (1.2,2.5) node (v) {violations};
\draw (v.0) edge[dotted,bend left=10,-latex',shorten <=2mm] (2.8,1.8);
\draw (v.0) edge[dotted,bend left,-latex',shorten <=2mm] (3.8,2.2);
\draw (v.0) edge[dotted,bend left,-latex',shorten <=2mm] (4.3,2.5);

\end{scope}
\end{tikzpicture}
\end{center}
\caption{Energy level and $\nbV$ along a winning run in a \LVenergynb reachability game.}
\label{fig_criteria}
\end{figure}
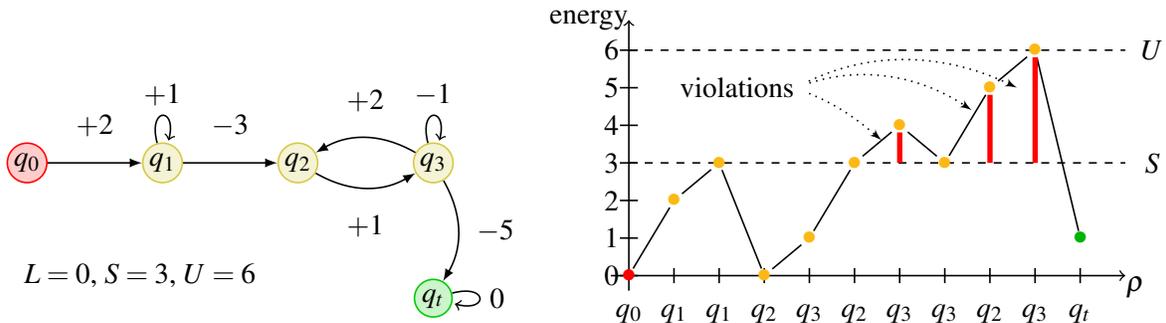




Given three values~$L\leq S \leq U$,
and  a~threshold~$V\in\bbN$,
the \LVenergynb (resp.~\LVenergyconsnb, \LVenergysum) objective is
the set of \LUfeasible infinite paths~$\pi$ such that,
along their associated runs~$\rho$ from~$(\qinit, L)$,
the~number~$\nbV(\rho)$ of
  violations (resp.~maximal number of consecutive
  violations~$\consnbV(\rho)$, sum~$\hV(\rho)$ of violations) of the
  soft upper bound~$S$ is at most~$V$.
%
  Similarly, for a set of states~$R$, the \LVenergynb-reachability
  (resp.~\LVenergyconsnb-reachability,
  \LVenergysum-reachability)  objective is the
  set of \LUfeasible paths~$\pi$ reaching~$R$ such that along their
  associated run from~$(\qinit, L)$, the number~$\nbV(\rho)$ of
  violations (resp.~maximal number of consecutive
  violations~$\consnbV(\rho)$, sum~$\hV(\rho)$ of violations) of the
  soft upper bound~$S$ is at most~$V$.


We~study the complexity of deciding the existence of a winning
strategy for the objectives defined above, in both the one- and
two-player settings. Further, for
\LVenergyall
games (for~$\star$ in $\{\#, \overline\#, \mathord\sum\}$),
we also address the following problems:
\begin{itemize}

\item {\bfseries bound existence:} Given $L$, $S$ and~$V$, decide if there exists a value $U\in \mathbb Z$ such that \Pl1 wins the \LVenergyall game;
\item {\bfseries minimization:} Given $L$ and~$S$, and a bound~$V_{\max}$,
  find a value $U\in \mathbb Z$ such that \Pl1 wins the game
with the least possible violations less than~$V_{\max}$, if~any.
\end{itemize}

Table~\ref{table-results} summarizes known
results, and the results obtained in this paper (where \LVenergyall
gathers all three energy constraints with violations). We furthermore show that the minimization problem for \LVenergyall (reachability) games require algorithms that run in \PSPACE in the one-player case, and in \EXPTIME in the two-player case.


\begin{table}[ht]
  \centering
  \def\arraystretch{1.3}
  \def\comp{c.}
  \scalebox{0.75}{%
  \begin{tabular}{|>{ }l<{ }||>{ }c<{ }|>{ }c<{ }|>{ }c|>{ }c|}
    \hline
    &\multicolumn{2}{c|}{Reachability}&\multicolumn{2}{c|}{Infinite runs}\\
    \hline
    & 1 player & 2 players & 1 player & 2 players
    \\\hline\hline
    \Lenergy
      &in \PTIME~(Thm.~\ref{thm_reachability_games}, \cite{Chatterjee0H17})
      & in \NP$\cap$ \co\NP~(Thm.~\ref{thm_reachability_games}, \cite{Chatterjee0H17})
    & in \PTIME~\cite{BouyerFLMS08} & in \NP $\cap$ \co\NP~\cite{BouyerFLMS08}\\\hline
    \LUenergy&\PSPACE-\comp~(Thm.~\ref{thm_reach_strong_LU})& \EXPTIME-\comp ~(Thm.~\ref{thm_reach_strong_LU})
    & \PSPACE-\comp~\cite{BouyerFLMS08} &\EXPTIME-\comp~\cite{BouyerFLMS08}\\\hline
    \LWenergy&in \PTIME~(Thm.~\ref{thm_1P_LW_reachability})& in \co\NP~(Coro.~\ref{cor_2P_LW_reachability}) 
    & in \PTIME~\cite{BouyerFLMS08} & in \NP $\cap$ \co\NP~\cite{BouyerFLMS08}\\\hline
    \LVenergyall &\PSPACE-\comp~(Thm.~\ref{thm_apna_LUHV_2P}) &\EXPTIME-\comp~(Thm.~\ref{thm_apna_LUHV_2P})&\PSPACE-\comp~(Thm.~\ref{thm_apna_LUHV_2P}) &\EXPTIME-\comp~(Thm.~\ref{thm_apna_LUHV_2P}) \\\hline
    %
    %
    Bound existence&\PSPACE-\comp~(Thm.~\ref{thm_Apnar_exists_min_exptimec})&\EXPTIME-\comp~(Thm.~\ref{thm_Apnar_exists_min_exptimec}) &\PSPACE-\comp~(Thm.~\ref{thm_Apnar_exists_min_exptimec})&\EXPTIME-\comp~(Thm.~\ref{thm_Apnar_exists_min_exptimec}) \\\hline
    %
    %
        Minimization&\PSPACE-\comp~(Thm.~\ref{thm_minimization})&\EXPTIME-\comp~(Thm.~\ref{thm_minimization})&\PSPACE-\comp~(Thm.~\ref{thm_minimization})&\EXPTIME-\comp~(Thm.~\ref{thm_minimization})\\\hline
    
  \end{tabular}}
  \caption{Summary of our results}
  \label{table-results}
\end{table}

\vspace{-\medskipamount}
\section{Energy reachability games with strict bounds}
\label{section_classical_energy_games}


\label{sec-single}


In this section, we focus on the \Lenergy-reachability
and \LUenergy-reachability problems. We~first prove that
\Lenergy-reachability problems are inter-reducible with \Lenergy
problems, which entails:
\begin{restatable}{theorem}{thmLLU}
\label{thm_reachability_games}
Two-player \Lenergy-reachability games are decidable in $\NP\cap\co\NP$.
The one-player version is in \PTIME.
\end{restatable}

\begin{remark}
It is worth noticing that these results are not a direct consequence of
the results of~\cite{ChatterjeeD12} about energy parity games:
in~that paper,
the~authors focus on the \emph{existence of an initial energy level}
for which \Pl1 has a winning strategy with energy-parity objectives
(which encompass our energy-reachability objectives). When the answer
is positive, they can compute the minimal initial energy level for
which a winning strategy exists, but the (deterministic) algorithm
runs in exponential time.
\end{remark}

\begin{remark}
These results were already proven in~\cite{Chatterjee0H17}: for
one-player arenas, the authors develop a \PTIME algorithm, while they
prove \LOGSPACE-equivalence with \Lenergy games for the two-player
setting (the~result then follow from~\cite{BouyerFLMS08}). Our proof
uses similar arguments as in the latter proof,
but with the same full and direct reductions back and forth both for
the one- and two-player cases.
\end{remark}

\begin{proof}
We prove that \Lenergy and \Lenergy-reachability
are \LOGSPACE-equivalent; our~reductions are valid both for one- and
two-player arenas. The~result then follows from~\cite{BouyerFLMS08}.

Intuitively, take a two-player \Lenergy game~$G$, and build the
arena~$G'$ as in Fig.~\ref{fig-redtoreachmain}: roughly, $G'$~is
obtained from~$G$ by increasing all weights by some small\footnote{The
value of~$\epsilon$ will generally be a rational, but it can be made
integer by scaling up all constants.}  positive value~$\epsilon$,
and \Pl1 is given the possibility to go to the target state, with a
large negative cost~$-\delta$, after each transition.  If~\Pl1 wins
in~$G$, playing her winning strategy in~$G'$ will grow the energy
level arbitrarily high, and will eventually allow her to take a
transition to~$q_t$. Conversely, if~\Pl1 wins in~$G'$, then she must
be able to grow the energy to above~$\delta$; taking~$\delta$ to be
larger than the sum of all positive weights in~$G$, \Pl1~must have a
strategy to force positive cycles before
reaching~$q_t$. If~$\epsilon$~is small enough, any positive cycle
in~$G'$ is a non-negative cycle in~$G$, so that the repeating same
strategy is winning in~$G$.

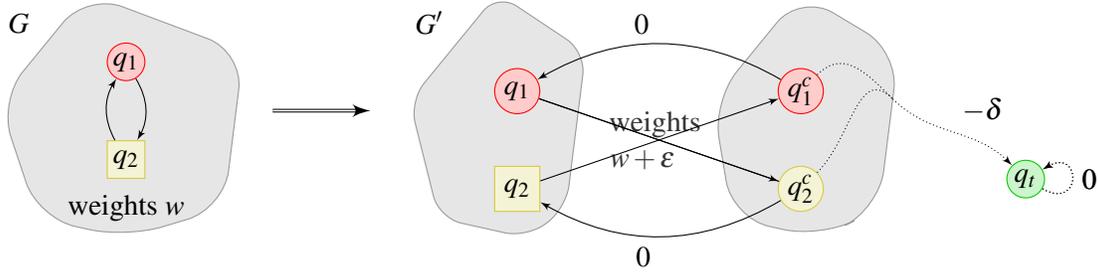
\begin{figure}[ht]
  \centering
  \begin{tikzpicture}[scale=1.3]
    \begin{scope}
    \draw (-.1,.9) node {$G$};
    \draw[rounded corners=5mm,grisc] (-.3,.1) -- (.6,1.2) -- (2.2,.6) -- (2,-1) -- (1,-1.3) -- (.1,-1.1) -- cycle;
      \draw (1,.5) node[rond5,rouge] (a) {} node {$q_1$};
       \draw (1,-.5) node[carre5,jaune] (c) {$q_2$};
      \draw (a) edge[-latex',bend left] (c); 
      \draw (c) edge[-latex',bend left] (a); 
      \path (1,-1) node {weights~$w$};
    \end{scope}

    \begin{scope}[xshift=5cm]
      \begin{scope}[xshift=-8mm]
      \draw (-.1,.9) node {$G'$};
        \draw[rounded corners=3mm,grisc] (-.3,.1) -- (.6,1.2) -- (1.5,.6) -- (1.3,-1) node[coordinate,pos=.3] (z) {} -- (.9,-1.3) -- (0,-1.1) --     cycle;
       \draw (.8,.2) node[rond7,rouge] (a) {} node {$q_1$};
       \draw (.8,-.8) node[carre6,jaune] (c) {} node {$q_2$};
      \end{scope}
       \draw (5.2,-.7) node[rond5,vert] (d) {} node {$q_t$};
      \draw (d) edge[-latex',out=-30,in=30,looseness=6,densely dotted] node[right] {$0$} (d);
    \draw[rounded corners=5mm,grisc] (2,.1) -- (2.9,1.2) -- (3.9,.6) -- (3.7,-1) node[coordinate,pos=.3] (z) {} -- (3.1,-1.3) -- (2.2,-1.1) --     cycle;

      \draw (2.9,.2) node[rond7,rouge] (a') {} node {$q_1^c$};
       \draw (2.9,-.8) node[rond7,jaune] (c') {} node {$q_2^c$};
      \draw (a) edge[-latex'] (c'); 
      \draw (c) edge[-latex'] (a'); 
      \path (a) edge node[opacity=.8,text width=1.3cm] {weights $w+\epsilon$} (c');
      \draw (a') edge[out=45,in=135,densely dotted,looseness=1.2] (z);
      \draw (c') edge[out=45,in=135,densely dotted,looseness=1.2] (z);
      \draw (c') edge[bend left,-latex'] node[below left]{$0$}(c);
      \draw (a') edge[bend right,-latex'] node[above left]{$0$}(a);
      \draw (z) edge[densely dotted,-latex',out=-45] node[above right] {$-\delta$} (d);
      \draw (d) edge[-latex',out=-30,in=30,looseness=6,densely dotted] node[right] {$0$} (d);
      \end{scope}
    \draw (2.5,0) edge[-latex',double] (3.5,0); 
  \end{tikzpicture}
  \vspace{-2\medskipamount}
  \caption{Schema of the reduction from \Lenergy to
    \Lenergy-reachability objectives}
    \label{fig-redtoreachmain}
\end{figure}
\begin{figure}[ht]
  \centering
  \begin{tikzpicture}[scale=1.3]
    \begin{scope}
    \draw (-.1,.9) node {$G$};
    \draw[rounded corners=5mm,grisc] (-.3,.1) -- (.6,1.2) -- (2.2,.6) -- (2,-1) -- (1,-1.3) -- (.1,-1.1) -- cycle;
      \path (1,-.8) node {weights $w$};
      \draw (.5,.5) node[rond5,jaune] (a) {} node {$q_i$};
      \draw (1.5,.4) node[rond5,jaune] (b) {};
      \draw (1.0,-.3) node[rond5,vert] (c) {} node {$q_t$};
      \draw (a.-135) edge[latex'-] +(-135:3mm);
      \draw (a) edge[-latex',bend left] (b);
      \draw (b) edge[-latex',bend left] (a);
      \draw (b) edge[-latex'] (c);
      \draw (c) edge[-latex'] (a);
      \draw (c) edge[-latex',out=-140,in=160,looseness=6] (c);
    \end{scope}
    \begin{scope}[xshift=6cm]
    \draw (-.1,.9) node {$G'$};
    \draw[rounded corners=5mm,grisc] (-.3,.1) -- (.6,1.2) -- (2.2,.6) -- (2,-1.2) node[coordinate,pos=.3] (z) {} -- (1,-1.3) -- (.1,-1.2) -- cycle;
      \path (1,-.9) node[text width=1.9cm,align=center] {weights $w-\epsilon$};
      \draw (.5,.5) node[rond5,jaune] (a) {} node {$q_i$};
      \draw (1.5,.4) node[rond5,jaune] (b) {};
      \draw (1.0,-.3) node[rond5,vert] (c) {} node {$q_t$};
      \draw (-1,-.4) node[rond5,rouge] (d) {} node {$q_i'$};
      \draw (d.-135) edge[latex'-] +(-135:3mm);
      \draw (a) edge[-latex',bend left] (b);
      \draw (b) edge[-latex',bend left] (a);
      \draw (b) edge[-latex'] (c);
      \draw (c) edge[-latex',out=10,in=-50,looseness=6,densely dotted] node[right] {$0$} (c);
      \draw (d) edge[-latex',densely dotted] node[pos=.4,above left=-1mm] {$1$} (a);
    \end{scope}
    \draw (3.2,0) edge[-latex',double] (4.2,0); 
  \end{tikzpicture}
  \vspace{-\medskipamount}
  \caption{Schema of the reduction from \Lenergy-reachability to
    \Lenergy objectives}
    \label{fig-redfromreach-main}
  \vspace{-\medskipamount}
\end{figure}
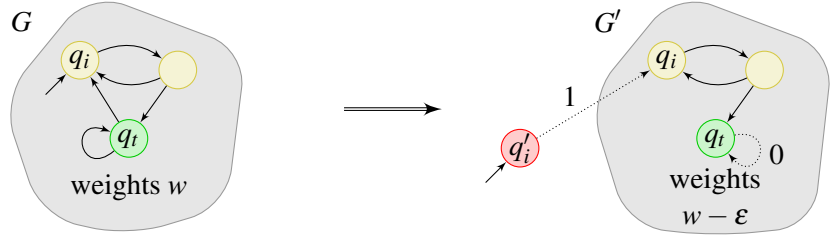

The~converse reduction is similar: this time, starting from a
two-player \Lenergy-reachability game~$G$, we~build a
two-player \Lenergy game~$G'$ by first restricting\footnote{This may
require to add a (losing) sink state, so that each state has an
outgoing transition.} to states from which \Pl1 has a strategy to
reach~$q_t$ (with no quantitative constraints), removing the outgoing
transitions of~$q_t$ and only keep a zero-self-loop, and substract a
small positive~$\epsilon$ to all other weights; to~compensate with
this, we~set the initial energy level to~$1$. Then if~\Pl1 wins
in~$G$, then by taking~$\epsilon$ small enough, the same strategy is
winning in~$G'$. Conversely, if she wins in~$G'$ (for \Lenergy
objective), then some of the outcomes may reach~$q_t$ and some other
will take non-negative cycles. When played in~$G$, that strategy will
reach~$q_t$ on some outcomes, and it will take positive cycles on some
others; along the latter outcomes, the energy will grow arbitrarily
high, and from some point on, \Pl1 will simply have to switch to her
attractor strategy in order to reach~$q_t$.
\end{proof}

Similarly, for \LUenergy-reachability objectives, we~prove the same
complexities as with classical \LUenergy objectives:
\begin{theorem}
\label{thm_reach_strong_LU}
One-player \LUenergy-reachability games are \PSPACE-complete.
Two-player \LUenergy-reachability games are \EXPTIME-complete.
\end{theorem}

\begin{proof}
Membership in \EXPTIME is proven by considering the \emph{expanded
  game}~\exGLU:
it~can be used to check reachability for both the one- and the two-player
cases. For~the one-player case, this can be achieved by proceeding
on-the-fly, without explicitly building the expanded game;
the~resulting algorithm runs in \PSPACE. For the two-player case,
we~solve reachability in that exponential-size game, which results in
an \EXPTIME algorithm.

For both the one- and the two-player settings, the hardness proofs
for \LUenergy objectives are readily adapted to \LUenergy-reachability
objectives, since they are based on reachability-like problems
(reachability in bounded one-counter automata~\cite{FearnleyJ13} and
countdown games~\cite{JurdzinskiLS07}, respectively).
\end{proof}


\vspace{-\medskipamount}
\section{Energy reachability games with weak upper bound}
\label{sec-weak}
%
\def\better{\mathrel{\triangleright}}

%

Finding a strategy that satisfies an \LWenergy constraint along an
infinite run is conceptually easy: it~suffices to find a cycle that
can be iterated once with a positive effect. It~follows that
memoryless strategies are enough, and the \LWenergy problem was shown
to be in \PTIME for one-player arenas, and in $\NP\cap\co\NP$ for
two-player arenas~\cite{BouyerFLMS08}.

The situation is different when we have a reachability condition:
players may have to keep track of the exact energy level in order to
find their way to the target state. Obviously, considering the
expanded arena~\exGLW, we~easily get exponential-time algorithms for
\LWenergy-reachability objectives. However, as proved below, in the one-player
case, a \PTIME algorithm exists. 


\begin{example}\label{ex-LW}
Consider the one-player arena of Fig.~\ref{fig-exLW}, where the lower
bound is~$L=0$, the weak-upper bound is~$W=5$, and the target state
is~$q_t$. Starting from~$q_0$ with initial credit~$0$, we~first have to move
to~$q_1$, and then iterate the positive cycle
$\beta_1=(q_1,+2,q_2)\cdot(q_2,-2,q3)\cdot(q_3,+1,q_1)$
three times, ending up in~$q_1$ with energy level~$3$. We~then take
the cycle $\beta_2=(q_1,+2,q_2)\cdot(q_2,-5,q4)\cdot(q_4,+5,q_1)$,
which raises the energy level to~$5$ when we come back to~$q_1$, so
that we can reach~$q_t$. Notice that $\beta_1$ has to be repeated
three times before taking cycle~$\beta_2$, and that repeating
$\beta_1$ more than three times maintains the energy level at~$4$,
which is not sufficient to reach~$q_t$. This suggests that \Pl1 needs
memory and cannot rely on a single cycle to win \LWenergy-reachability
games.
\end{example}

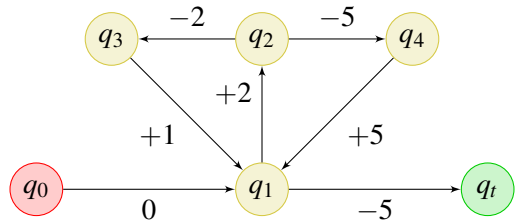
\begin{figure}[htbp]
  \centering
  \begin{tikzpicture}
    \draw (0,0) node[rond,rouge] (a) {} node {$q_0$};
    \draw (3,0) node[rond,jaune] (b) {} node {$q_1$};
    \draw (3,2) node[rond,jaune] (c) {} node {$q_2$};
    \draw (1,2) node[rond,jaune] (d) {} node {$q_3$};
    \draw (5,2) node[rond,jaune] (e) {} node {$q_4$};
    \draw (6,0) node[rond,vert] (g) {} node {$q_t$};
    \draw (a) edge[-latex'] node[below] {$0$} (b)
    (b) edge[-latex'] node[above left] {$+2$} (c)
    (c) edge[-latex'] node[above] {$-2$} (d)
    (d) edge[-latex'] node[below left] {$+1$} (b)
    (c) edge[-latex'] node[above] {$-5$} (e)
    (e) edge[-latex'] node[below right] {$+5$} (b)
    (b) edge[-latex'] node[below] {$-5$} (g);    
      \end{tikzpicture}
  \caption{A one-player arena with \LWenergy-reachability objective}
  \label{fig-exLW}
\end{figure}


This example shows that winning strategies for \Pl1 may have to
monitor the exact energy level all along the computation, thereby
requiring exponential memory (assuming that all constants are encoded
in binary; with unary encoding, the expanded game~\exGLW would have
polynomial size and directly give a polynomial-time algorithm).

\begin{proposition}
In \LWenergy-reachability games, exponential memory may be necessary
for \Pl1 (assuming binary-encoded constants).
\end{proposition}

Notice that this does not prevent from having \PTIME algorithms:
the~strategy in Example~\ref{ex-LW} is not very involved, but it
depends on the energy level (up~to~$W$). 

From now on, and until Lemma~\ref{lemma-higherstrat}, we~consider
the one-player case.  In order to get a polynomial-time algorithm,
we analyze cycles in the graph, and prove that a path
witnessing \LWenergy reachability can have a special form, which can
be represented compactly using polynomial size. We begin with a series
of simple lemmas. 

\begin{restatable}{lemma}{lemmaeight}
\label{lemma-higherrun}
Let $\pi$ be a finite path in a one-player arena~$G$. 
If $(q,u) \LWarrow\pi (q',u')$, then for any~$v\geq u$,
$(q,v) \LWarrow\pi (q',v')$ for some~$v'\geq u'$.
\end{restatable}

This lemma states that starting with higher energy level can only be beneficial.
Notice that, even if we add condition~$u'>u$ in the
hypotheses of Lemma~\ref{lemma-higherrun}, it~need not be the case that
$v'>v$. In~other terms, a~sequence of transitions may have a positive
effect on the energy level from some configuration, and a negative
effect from some other configuration, due to the weak upper bound.
Below, we prove a series of results related to this issue, and that
will be useful for the rest of the proof. First, the~effect of a given
path (i.e., the net amount of energy that is harvested) decreases when
the initial energy level increases:
\begin{restatable}{lemma}{lemmanine}
\label{lemma-hitW}
  Let $\pi$ be a finite path in a one-player arena~$G$, and consider
  two \LWruns $(q,u)\LWarrow\pi (q',u')$ and $(q,v)\LWarrow\pi(q',v')$
  with $u\leq v$. Then $u'-u\geq v'-v$, and if the inequality is
  strict, then the energy level along the run
  $(q,v)\LWarrow\pi(q',v')$ must have hit~$W$.
\end{restatable}

The next lemma is more precise about the effect of following a path
when starting from the maximal energy level~$W$:
\begin{restatable}{lemma}{lemmaten}
\label{lemma-W}
  Let $\pi$ be a finite path in a one-player arena~$G$, for which there
  is an \LWrun $(q,u)\LWarrow\pi (q',u')$. 
  If $u'$ is the maximal energy level along that run, then
  $(q,W)\LWarrow\pi (q',W)$;
  if $u$ is the maximal energy level along the run above, then
  $(q,W)\LWarrow\pi (q',W+u'-u)$.
\end{restatable}

From Lemma~\ref{lemma-higherrun}, it~follows that any run witnessing
\LWenergy reachability can be assumed to contain no cycle with
nonpositive effect. Formally:

\begin{restatable}{lemma}{lemmatwelve}
\label{lemma-removeneg}
  Let $\pi$ be a finite path in a one-player arena~$G$.
  If $(q,u) \LWarrow{\pi} (q',u')$ and $\pi$ can be decomposed as
  $\pi_1\cdot\pi_2\cdot \pi_3$ in such a way that
  $(q,u) \LWarrow{\pi_1}
  (s,v)\LWarrow{\pi_2}(s,v') \LWarrow{\pi_3}(q',u')$ with $v'\leq
  v$, then $(q,u)\LWarrow{\pi_1\cdot\pi_3} (q',u'')$ with $u''\geq
  u'$.
\end{restatable}

The following lemmas show that several occurrences of a cycle having
positive effect along a path can be gathered together. This will be
useful to prove the existence of a short 
path witnessing \LWenergy reachability.
\begin{restatable}{lemma}{lemmaeleven}
\label{lemma-iteratepos}
  Let $\pi$ be a finite path in a one-player arena~$G$.
  If
  $(q,u)\LWarrow\pi (q',u')$ with $u'> u$
  and 
  $(q,w)\LWarrow\pi (q',w')$ with $w'> w$,
  then
  for any $u\leq v\leq w$, it~holds that
  $(q,v)\LWarrow\pi (q',v')$ with $v'> v$.
\end{restatable}

\begin{restatable}{lemma}{lemmathirteen}
\label{lemma-iteratecycles}
  Let $\pi$ be a cycle on~$q$ such that $(q,u) \LWarrow{\pi} (q,v)$
  for some $u\leq v$. Then $(q,u) \LWarrow{\pi^{W-L}} (q,v')$ for
  some~$v' \geq v$, and $(q,v')\LWarrow{\pi} (q,v')$.
\end{restatable}

Fix a path~$\pi$ in~$G$, and assume that some cycle~$\phi$ appears
(at~least) twice along~$\pi$: the~first time from some
configuration~$(q,u)$ to some configuration~$(q,u')$, and the second
time from~$(q,v)$ to~$(q,v')$. First, we~may assume that $\phi$ has
length at most~$\size Q$, since otherwise we can take an inner
subcycle.  We~may also assume that $v>u'$, as otherwise we~can apply
Lemma~\ref{lemma-removeneg}
to get rid of the resulting nonpositive
cycle between~$(q,u')$ and~$(q,v)$. For the same reason we may assume
$u'>u$ and $v'>v$.  As~a consequence, by Lemmas~\ref{lemma-iteratepos}
and~\ref{lemma-higherrun},
by repeatedly 
iterating~$\phi$
from~$(q,u)$, we~eventually reach some
configuration~$(q,w)$ with $w\geq v'$, from which we can follow the
suffix of~$\pi$ after the second occurrence of~$\phi$. It~follows that
all~occurrences of~$\phi$ along~$\pi$ can be grouped together, and
we~can restrict our attention to runs of the form $\alpha_1\cdot
\phi_1^{n_1}\cdot \alpha_2\cdot\phi_2^{n_2}\cdots
\phi_k^{n_k}\cdot\alpha_{k+1}$ where the cycles~$\phi_j$ are distinct,
and have size at most~$\size Q$, and the finite runs~$\alpha_j$ are
acyclic.  Notice that each occurrence of any cycle~$\phi_j$ can be assumed to have positive effect, and by Lemma~\ref{lemma-iteratecycles}, we~may
assume $n_j=W-L$ for all~$j$.

While this allows us to only consider paths of a special form, this
does not provide \emph{short} witnesses, since there may be
exponentially many cycles of length less than or equal to~$\size Q$,
and the witnessing run may need to iterate several cycles looping on
the same state (see~Example~\ref{ex-LW}).
In~order to circumvent this problem, we have to show that all cycles need not be considered, and that one 
can compute the "useful" cycles efficiently. 
For this, we 
introduce \emph{universal} cycles, which are cycles
that can be iterated from any initial energy level (above~$L$).

\begin{definition}
Let $G$ be a one-player arena, and $q$~be a state of~$G$.  Let~$W$ be a weak-upper bound and~$L\leq
W$ be a lower bound. A~\emph{universal cycle}
on~$q$ in~$G$ is a cycle~$\phi$ with $\first(\phi)=\last(\phi)=q$ such
that $(q,L)\LWarrow{\phi} (q,v_{\phi,L})$ for some~$v_{\phi,L}$ (i.e.,
the~energy level never drops below the lower bound~$L$ when
following~$\phi$ with initial energy level~$L$).
A~universal cycle is \emph{positive} if ${v_{\phi,L}>L}$.
\end{definition}

When a cycle $\phi$ is iterated $W-L$ times in a row, then some universal
cycle $\sigma$ is also iterated $W-L-1$ times (by~considering the state with
minimal energy level along~$\phi$).  As~a consequence, iterating only
universal cycles is enough: we~may now only look for runs of the form
$\beta_1\cdot \sigma_1^{n_1}\cdot \beta_2\cdot\sigma_2^{n_2}\cdots
\sigma_k^{n_k}\cdot\beta_{k+1}$ where ~$\sigma_j$'s are \emph{universal} cycles
of length at most~$\size Q$. Now, assume that some state~$q$ admits
two universal cycles~$\sigma$ and~$\sigma'$, and that both cycles
appear along a given run~$\pi$. Write~$h$ (resp.~$h'$) for the energy
levels reached after iterating~$\sigma$ (resp.~$\sigma'$) $W-L$
times. We~define a preorder on universal cycles of~$q$ by letting
$\sigma \better \sigma'$ when $h\geq h'$.  Then if~$\sigma \better
\sigma'$, each occurrence of~$\sigma'$ along~$\pi$ can be replaced
with~$\sigma$, yielding a run~$\pi'$ that still satisfies the \LWenergy
condition (and~has the same first and last states).
Generalizing this argument, each state that admits universal cycles has an optimal universal cycle
of length at most~$\size Q$, and it is enough to iterate only this
universal cycle to find a path witnessing reachability. This provides
us with a \emph{small witness}, of the form $\gamma_1\cdot
\tau_1^{W-L}\cdot \gamma_2\cdot\tau_2^{W-L}\cdots
\tau_k^{W-L}\cdot\gamma_{k+1}$ where $\tau_j$ are optimal universal
cycles of length at most~$\size Q$ and $\gamma_j$ are acyclic
paths. Since it~suffices to consider at most one universal cycle per
state, we~have $k\leq\size Q$.
From this, we~immediately derive an \NP algorithm for solving
\LWenergy reachability for one-player arenas: it~suffices to
non-deterministically select each portion of the path, and compute
that each portion is \LWfeasible (notice that there is no need for
checking universality nor optimality of cycles; those properties were
only used to prove that small witnesses exist). Checking
\LWfeasibility requires computing the final energy level reached after
iterating a cycle $W-L$ times; this can be performed by detecting the
highest energy level along that cycle, and computing how much the
energy level decreases from that point on until the end of the
cycle. This~provides us with a way of \emph{accelerating} the computation 
of the effect of iterating cycles.

\medskip

We~now prove that optimal universal cycles of length at most~$\size Q$
can be computed for a given state~$q_0$. For~this we unwind the graph
from~$q_0$ as a DAG of depth~$\size Q$, so that it includes all cycles
of length at most~$\size Q$. We~name the states of this DAG $[q',d]$,
where $q'$ is the name of a state of
the arena and $d$ is the depth of this state in the DAG
(using square brackets to avoid confusion with configurations~$(q,l)$
where $l$ is the energy level); hence there are
transitions $([q',d],w,[q'',d+1])$ in the DAG as soon as there
is a transition $(q',w,q'')$ in the arena.

We~then explore this DAG from its initial state~$[q_0,0]$,
looking for (paths corresponding~to) universal cycles. Our~aim is to
keep track of all runs from~$[q_0,0]$ to~$[q',d]$ that are prefixes of
universal cycles starting from~$q_0$. Actually, we~do not need to keep track of
those runs explicitly, and it suffices for each such run to remember
the following two values:
\begin{itemize}
\item the maximal energy level~$M$ that has been observed along the
  run so~far (starting from energy level~$L$, with weak upper bound~$W$);
\item the difference~$m$ between the maximal energy level~$M$ and the
  final energy level in~$[q',d]$. Notice that~$m\geq 0$, and that the
  final energy level in~$[q',d]$ is~$M-m$. 
\end{itemize}

\begin{example}
Figure~\ref{fig-ex_cycles} shows two universal cycles from~$q_0$ in
an \LWenergy game with $L=0$ and $W=5$. The~first cycle, going
via~$q_2$, ends with $M_1=5$ (reached in~$q_2$) and $m_1=4$, thus with
a final energy level of~$1$ (when starting from energy level~$0$);
actually, iterating this cycle will not improve this final energy
level. The~second cycle, via~$q_4$ and~$q_5$, has a maximal energy
level $M_2=4$ (reached in~$q_3$) and ends with $m_2=1$. Hence, after
one iteration of this cycle, one can end in state $q_0$ with energy
level $\wub-m_2=4$.
\begin{figure}[htbp]
\centering

  \begin{tikzpicture}
  \begin{scope}
    \draw (0,0) node[rond,rouge] (q0) {} node {$q_0$};
    \draw (2,0) node[rond,jaune] (q1) {} node {$q_1$};
    \draw (2,1.5) node[rond,jaune] (q2) {} node {$q_2$};
    \draw (0,1.5) node[rond,jaune] (q3) {} node {$q_3$};
    \draw (4,0) node[rond,jaune] (q4) {} node {$q_4$};
     \draw (2,3) node[rond,jaune] (q5) {} node {$q_5$};
    \draw (q0) edge[-latex'] node[below] {$+2$} (q1)
    (q1) edge[-latex'] node[left] {$+3$} (q2)
    (q2) edge[-latex'] node[above] {$-3$} (q3)
    (q3) edge[-latex'] node[left] {$-1$} (q0)
    (q1) edge[-latex'] node[above] {$-1$} (q4)
    (q4) edge[-latex'] node[right] {$+2$} (q5)
    (q5) edge[-latex', bend right] node[below] {$+1$} (q3);  
    
    \end{scope}
    \begin{scope}[xshift=6cm,yscale=0.7]
\draw[->] (-0.2,0) -- (5.5,0);
\draw[->] (0,-0.2) -- (0,4.5);

\foreach  \x in {1,2,3,4,5}{
\draw[] (\x,-0.1) node[below] {$\x$} -- (\x,0.1){};
}
\foreach  \y/\t in {0.8/1,1.6/2,2.4/3,3.2/4,4/5}{
\draw[] (-0.1,\y) node[left] {$\t$} -- (0.1,\y){};
}


\draw[dashed] (0,4) -- (5.5,4){};
\node (lh) at (-1,4) {$\wub=$};

\tikzstyle{noeud}=[fill=white,draw,circle,inner sep=0pt,minimum width=4mm]

\node[] (l0) at (0,-0.5) {$0$};
\node[noeud,rouge] (p0) at (0,0) {$q_0$};

\node[noeud, jaune] (p1) at (1,1.6) {$q_1$};
\draw[] (p0)--(p1){};

\node[noeud, jaune] (p2) at (2,4) {$q_2$};
\draw (p1)--(p2){};

\node[noeud, jaune] (p3) at (3,1.6) {$q_3$};
\draw (p2)--(p3){};

\node[noeud,rouge] (p4) at (4,0.8) {$q_0$};
\draw (p3)--(p4){};

\node[noeud, jaune] (p5) at (2,0.8) {$q_4$};
\draw [color=red]  (p1)--(p5){};

\node[noeud, jaune] (p6) at (3,2.4) {$q_5$};
\draw [color=red] (p5)--(p6){};

\node[noeud, jaune] (p7) at (4,3.2) {$q_3$};
\draw [color=red] (p6)--(p7){};

\node[noeud,rouge] (p8) at (5,2.4) {$q_0$};
\draw [color=red] (p7)--(p8){};

\node[] (labs) at (6,-0.5) {$depth$};

\node[] (lord) at (-0.8,4.5) {$energy$};

\end{scope}

 \end{tikzpicture}
  \caption{Two cycles with upper bound $\wub=5$}
  \label{fig-ex_cycles}

\end{figure}
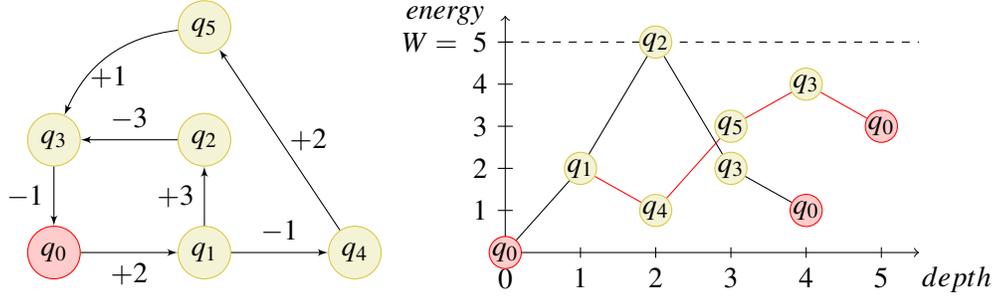

\end{example}

If~we know the values~$(M,m)$ of some path from~$[q_0,0])$
to~$[q',d]$, we~can decide if a given transition with weight~$w$
from~$[q',d]$ to~$[q'',d+1]$ can be taken (the resulting path can
still be a prefix of a universal cycle if $M-m+w \geq L$ ), and how
the values of~$M$ and~$m$ have to be updated: if $w>m$, the~run will
reach a new maximal energy level, and the new pair of values is
$(\min(W;M-m+w),0)$; if $m+L-M\leq w\leq m$, then the transition can
be taken: the new energy level~$M-m+w$ will remain between~$L$
and~$M$, and we update the pair of values to~$(M,m-w)$; finally, if
$w<m+L-M$, the energy level would go below~$L$, and the resulting run
would not be a prefix of a universal cycle.

Following these ideas, we inductively attach labels to the states of the DAG:
initially, $[q_0,0]$ is labelled with~$(M=L,m=0)$; then if a
state~$[q',d]$ is labelled with~$(M,m)$, and if there is a transition
from~$[q',d]$ to~$[q'',d+1]$ with weight~$w$:
\begin{itemize}
\item if $w>m$, then we~label $[q'',d+1]$ with the pair $(\min(W;M-m+w),0)$;
\item if $m+L-M\leq w\leq m$, we label $[q'',d+1]$ with $(M,m-w)$.
\end{itemize}

The following lemma makes a link between runs in the DAG and labels
computed by our algorithm:
\begin{restatable}{lemma}{lemmasixteen}
\label{lemma-DAGlabel}
  Let~$[q,d]$ be a state of the DAG, and $M$ and~$m$ be two integers
  such that $0\leq m\leq M-L$.  Upon termination of this algorithm, 
  state~$[q,d]$ of the DAG is labelled with~$(M,m)$ if, and only~if,
  there is an \LWrun of length~$d$ from~$(q_0,L)$ to~$(q,M-m)$ along
  which the energy level always remains in the interval~$[L,M]$ and
  equals~$M$ at some point.
\end{restatable}

%

\begin{restatable}{lemma}{lemmaseventeen}
\label{lemma-univcycle}
  Let~$[q_0,d]$ be a state of the DAG, with~$d>0$. Let $m$ be a
  nonnegative integer such that $L+m<W$.  Upon termination of this algorithm,
  state~$[q_0,d]$ is labelled with~$(M,m)$ for some $M>L+m$
  if, and only if, there is a universal cycle~$\phi$ on~$q_0$
  of length~$d$ such that $(q_0,L) \LWarrow{\phi^{W-L}} (q_0,W-m)$.
\end{restatable}

\begin{proof}
First assume that~$[q_0,d]$ is labelled with~$(M,m)$ for some~$M$ such
that $M-m>L$. From Lemma~\ref{lemma-DAGlabel}, there is a cycle~$\phi$
on~$q_0$ of length~$d$ generating a run $(q_0,L)\LWarrow{\phi}(q_0,M-m)$
along which the energy level is within~$[L,M]$.  Then~$M-m\geq L$, so
that Lemma~\ref{lemma-iteratecycles} applies: we~then get $(q_0,L)
\LWarrow{\phi^{W-L}} (q_0,E)$ with $(q_0,E)\LWarrow\phi(q_0,E)$ and $E\geq L$.
Write $(p_i)_{0\leq i<\size\phi}$ for the sequence of weights along~$\phi$.
Also write $\rho$ for the run $(q_0,L)\LWarrow{\phi}(q_0,M-m)$, and
$\sigma$ for the run $(q_0,E)\LWarrow{\phi}(q_0,E)$.

As $L<M-m$, then by
Lemma~\ref{lemma-hitW}, it~must be the case that energy level~$W$ is
reached along~$\sigma$.  Write~$i_0$ for the first
position along~$\rho$ for which the energy level
is~$M$.
Assume $\tilde\sigma_{i_0}\not=W$: by~Lemma~\ref{lemma-higherrun},
we~must have $M=\tilde\rho_{i_0}\leq \tilde\sigma_{i_0}<W$.  Then for
all~$k\geq i_0$, $\sum_{l=i_0}^k p_l\leq 0$.
Since $\tilde\sigma_{i_0}<W$, then also $\tilde\sigma_{k}<W$
for all~$k\geq i_0$.
According to Lemma~\ref{lemma-hitW}, energy level~$\wub$ is reached in~$\sigma$, so 
there exists some~$k_0<i_0$ such that $\tilde\sigma_{k_0}=W$.
However, as $i_0$ is the index of the first maximal value in~$\rho$, we have $\tilde\rho_{k_0}<M$, and the energy level increases in run $\rho$ between $k_0$ and $i_0$. So according to Lemma~\ref{lemma-W}, we should have $\tilde\sigma_{i_0}=W$, which raises a contradiction.
Hence we proved $\tilde\sigma_{i_0}=W$; applying
the second result of Lemma~\ref{lemma-W}, we~get ${E=W-m}$.

\medskip
Conversely, if there is a universal
cycle~$\phi$ satisfying the
conditions of the lemma,
then it must have positive effect when run from energy level~$L$.
Let~$F$ be such that $(q_0,L)\LWarrow{\phi}
(q_0,F)$, and~$M$ be the maximal energy level encountered along
the run $(q_0,L)\LWarrow{\phi} (q_0,F)$. By
Lemma~\ref{lemma-DAGlabel}, state~$[q_0,d]$ is labelled with~$(M,m')$
for some~$m'\geq 0$ such that $F=M-m'$.
By~Lemma~\ref{lemma-iteratecycles}, we~must have $(q_0,L)
\LWarrow{\phi^{W-L}} (q_0,W-m')$.
\end{proof}

The algorithm above computes optimal universal cycles, but it still
runs in exponential time (in the worst case) since it may generate
exponentially many different labels in each state~$[q,d]$ (one per
path from~$[q_0,0]$ to~$[q,d]$). We~now explain how to only generate
polynomially-many pairs~$(M,m)$. This is based on the following
partial order on labels: we~let $(M,m) \preceq (M',m')$ whenever
$M-m\leq M'-m'$ and $m'\leq m$. Notice in particular that
\begin{itemize}
\item if~$M=M'$, then $(M,m)\preceq (M',m')$ if, and only~if, $m'\leq m$;
\item if $m=m'$, then $(M,m)\preceq (M',m')$ if, and only~if, $M\leq M'$.
\end{itemize}
The following lemma entails that it suffices to store maximal labels w.r.t. $\preceq$:

\begin{restatable}{lemma}{lemmaeighteen}
\label{lemma-order-cycle}
Consider two paths~$\pi$ and~$\pi'$ such that
$\first(\pi)=\first(\pi')$ and $\last(\pi)=\last(\pi')$, and with
respective values~$(M,m)$ and~$(M',m')$ such that $(M,m)\preceq (M',m')$.
If~$\pi$ is a prefix of a universal cycle~$\phi$, then $\pi'$~is a
prefix of a universal cycle~$\phi'$ with $\phi'\better\phi$.
\end{restatable}

It~remains to prove that by keeping only maximal labels, we only store
a polynomial number of labels:

\begin{restatable}{lemma}{lemmanineteen}
\label{lemma-DAG-construction-poly}
If the algorithm labelling the DAG only keeps maximal labels (for~$\preceq$), then
it runs in polynomial time.
\end{restatable}

\begin{proof}
We prove that,
when attaching to each node $[q,d]$ of the DAG only the maximal labels (w.r.t $\preceq$) reached for a path of length $d$ ending in state $q$,
the number of
values for the first component of the different labels that appear at
depth~$d>0$ in the DAG is at most $d\cdot \size Q$. Since it only
stores optimal labels, our algorithm will never associate to a state $[q,d]$
two labels having the same value on their first component. So, any
state at depth~$d$ will have at most $d\cdot \size Q$ labels.

So we prove, by induction on~$d$, that the number of different values
for the first component among the labels appearing at depth~$d>0$ is
at most~$d\cdot\size Q$. 
This is true for~$d=1$ since the initial state~$(q,0)$ only
contains~$(M=0,m=0)$, and each transition with nonnegative weight~$w$
will create one new label~$(w,0)$ (transitions with negative weight
are not prefixes of universal cycles). Now, since all those labels
have value~$0$ as their second component, each state $[q,1]$ in the
DAG will be attached at most one label. Hence, the total number of
labels (and the total number of different values for their first
component) is at most~$\size Q$ at depth~$1$ in the~DAG.

Now, assume that labels appearing 
at depth $d>1$
are all drawn from a set of labels~$\{(M_i,m_i) \mid 1\leq i\leq n\}$
in which the number of different values of~$M_i$ is at most
$d\cdot\size Q$.  Consider a state~$[q',d]$, labelled with
$\{(M_i,m_i) \mid 1\leq i\leq n_{q',d}\}$ (even if it means reindexing
the labels). Pick~a transition from~$[q',d]$ to~$[q'',d+1]$, with
weight~$w$. For each pair~$(M_i,m_i)$ associated with $[q',d]$,
it~creates a new label 
in~$[q'',d+1]$, which is
\begin{itemize}
\item either $(\min(W;M_i-m_i+w),0)$ if~$m_i<w$ (maximal energy level increases); 
\item or $(M_i,m_i-w)$ if $m_i+L-M_i\leq w\leq m_i$ (maximal energy level in unchanged).
\end{itemize}
Now, for a state~$(q'',d+1)$, the~set of labels created by all
incoming transitions can be grouped as follows:
\begin{itemize}
\item labels having~zero as their second component; among those, our
  algorithm only stores the one with maximal first component, as
  $(M_i,0) \preceq (M_j,0)$ as soon as $M_i\leq M_j$;
\item for each $M_i$ appearing at depth~$d$, labels having~$M_i$ as
  their first component; again, we only keep the one with minimal
  second component, as $(M,m_i) \preceq (M,m_j)$ when $m_j \leq m_i$.
\end{itemize}
  Last, for this state~$[q'',d+1]$, we keep at most one label
  for each distinct value among the first components $M_i$ of labels
  appearing at depth~$d$, and possibly one extra label with second
  value~$0$. In~other terms, at depth $d+1$ the values that appear as
  first component of labels are obtained from values at depth~$d$,
  plus possibly one per state; Hence, at depth $d+1$, there exists at
  most $(d+1)\cdot\size Q$ labels, which completes the proof of the
  induction step.
\end{proof}

\bigskip

Using the algorithm above, we can compute, for each state~$q$ of the
original arena, the~smallest value~$m_q$ for which there exists a
universal cycle on~$q$ that, when iterated sufficiently many times,
leads to configuration~$(q,W-m_q)$. Since universal cycles can be
iterated from any energy level, if $q$ is reachable, then it is
reachable with energy level~$W-m_q$. We~make this explicit by adding
to our arena a special self-loop on~$q$, labelled with~$\set(W-m_q)$,
which sets the energy level to~$W-m_q$ (in~the same way as
\emph{recharge transitions} of~\cite{EjsingDuun2013InfiniteRI}).

In~the resulting arena, 
we can restrict to paths of the
form $\gamma_1\cdot \nu_1\cdot \gamma_2\cdot\nu_2\cdots
\nu_k\cdot\gamma_{k+1}$, where $\nu_i$ are newly added transitions
labelled with~$\set(W-m)$, and $\gamma_i$ are acyclic paths. Such
paths have length at most~$(\size Q+1)^2$. We~can then inductively
compute the maximal energy level that can be reached (under our
\LWenergy constraint) in any state after paths of length less than or
equal to~$(\size Q+1)^2$.  This can be performed by unwinding
(as~a~DAG) the modified arena from the source state~$\qinit$ up to
depth~$(\size Q+1)^2$, and labelling the states of this DAG by the
maximal energy level with which that state can be reached
from~$(\qinit,L)$; this is achieved in a way similar to our
algorithm for computing the effect of universal cycles, but this time
only keeping the maximal energy level that can be reached (under
\LWenergy constraint). As there are at most $\size Q$ states per level in this DAG of depth at most $(\size Q+1)^2$, 
we~get:

\begin{theorem}
\label{thm_1P_LW_reachability}
The existence of a winning path in one-player \LWenergy-reachability games
can be decided in \PTIME.
\end{theorem}

\begin{example}  
Consider the one-player arena of Fig.~\ref{fig-ex}. We~assume $L=0$, and
fix an even  weak upper bound~$W$. 
The~state~$s$ has $W/2$ disjoint cycles: for~each odd integer~$i$ in
$[0;W-1]$, the~cycle~$c_i$ is made of three consecutive edges with
weights $-i$, $+W$ and $-W+i+1$. Similarly, the~state~$s'$ has $W/2$
disjoint cycles: for even integers~$i$ in~$[0;W-1]$, the cycle~$c'_i$
has weights $-i$, $+W$ and~$-W+i+1$.  Finally, there are: two sequences of $k$ edges of weight~$0$ from~$s$ to~$s'$ and
  from~$s'$ to~$s$; an edge from the initial state to~$s$ of weight~$1$, and from~$s'$ to target state  $q_t$ of weight $-W$. The total number of states then is $2W+2k+2$.


In order to go from the initial state~$q_0$, with energy level~$0$, to the
final state~$q_t$, we have to first take the cycle $c_1$ (with weights~$-1$,
$+W$, $-W+2$) on~$s$ (no~other cycles~$c_i$ can be taken). We~then
reach configuration~$(s,2)$. Iterating~$c_1$ has no effect, and
the only next interesting cycle is~$c_2$, for which we~have to go
to~$s'$. After running~$c_2$, we~end up in~$(s',3)$. Again,
iterating~$c_2$ has no effect, and we go back to~$s$, take~$c_3$, and
so~on. We~have to take each cycle~$c_i$ (at~least) once, and take the
sequences of~$k$ edges between~$s$ and~$s'$ $W/2$ times each.
In~the~end, we~have a run of length $3W+Wk+2$. 

\begin{figure}[h]
  \centering
  \begin{tikzpicture}
    \draw (0,1.5) node[rond,rouge] (ini) {} node {$q_0$};
    \draw (5,1.5) node[rond,rouge] (tgt) {} node {$q_t$};
    \draw (0,0) node[rond,vert] (a) {} node {$s$};
    \draw (5,0) node[rond,vert] (b) {} node {$s'$};
    \draw (ini) edge[-latex'] node[right] {$1$} (a);
    \draw (b) edge[-latex'] node[right] {$-W$} (tgt);
    \draw (a) edge[bend left,-latex'] node[above] {$k$ edges} node[below] {weight=$0$} (b);
    \draw (b) edge[bend left,-latex'] node[above] {$k$ edges} node[below] {weight=$0$} (a);
    \draw (a) edge[out=160,in=200,looseness=9,-latex'] node[left] {$\genfrac{}{}{0pt}{0}{-i;+W;-W+i+1}{(\text{$i$ odd})}$} (a);
    \draw (b) edge[out=-20,in=20,looseness=9,-latex'] node[right] {$\genfrac{}{}{0pt}{0}{-i;+W;-W+i+1}{(\text{$i$ even})}$} (b);
  \end{tikzpicture}
  \vspace{-2\medskipamount}
  \caption{An example showing that more than one cycle per state can be needed.}\label{fig-ex}
  \vspace{-\medskipamount}
\end{figure}
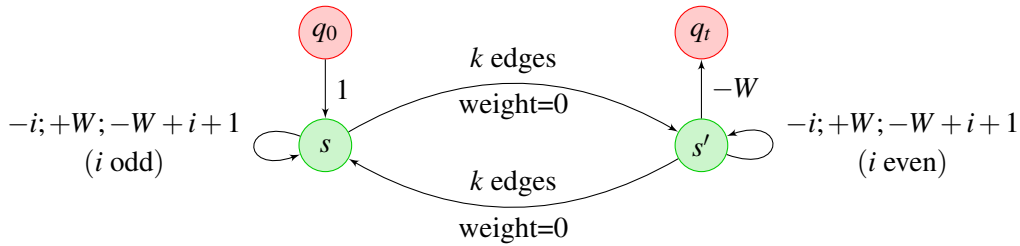

Let us look at the universal cycles that we have in this arena:
besides the cycles made of the $2k$ edges with weight~zero between~$s$
and~$s'$, the only possible universal cycles have to depart from the
first state of each cycle~$c_i$ (as~they are the only states having a
positive outgoing edge). Following Lemma~\ref{lemma-iteratecycles},
such cycles can be iterated
arbitrarily many times, and set the energy level to some value
in~$[L;W]$. Since the only edge available at the end of a universal
cycle has weight~$+W$, the~exact value of the universal cycles is
irrelevant: the~energy level will be~$W$ anyway when reaching the
second state of each cycle~$c_i$. As~a consequence, using \set-edges
in this example does not shorten the witnessing run, which then cannot
be shorter than~$3W+Wk+2$.
\end{example}

We now move to the two-player setting. We~begin with proving a result
similar to Lemma~\ref{lemma-higherrun}:
\begin{restatable}{lemma}{lemmatwtwo}
\label{lemma-higherstrat}
  Let $G$ be a two-player arena, equipped with an \LWenergy-reachability
  objective. Let~$q$ be a state of~$G$, and~$u\leq u'$ in $[L;W]$. If
  \Pl1 wins the game from~$(q,u)$, then she also wins from~$(q,u')$.
\end{restatable}

By Martin's theorem~\cite{Mar75}, our games are determined. It~follows that if \Pl2 wins from
some configuration~$(q,v)$, she also wins from~$(q,v')$ for
all~$L\leq v'\leq v$ (assuming the contrary, i.e. $(q,v')$ winning for \Pl1, would lead to the contradictory statement that $(q,v)$ is both winning for \Pl1 and \Pl2). Using classical techniques~\cite{ChatterjeeD12},
we prove that
\Pl2 can be restricted to play
memoryless strategies:

\begin{restatable}{proposition}{propXXIV}
\label{mem-proposition}
For two-player \LWenergy-reachability games,
memoryless strategies are sufficient
for \Pl2.
\end{restatable}

A direct consequence of this result and of Theorem~\ref{thm_1P_LW_reachability} is the following:
\begin{corollary}
\label{cor_2P_LW_reachability}
  Two-player \LWenergy-reachability games are in $\co\NP$. 
\end{corollary}



%

\vspace{-\medskipamount}
\section{Energy reachability games with soft upper bound}
\label{section_APNA}

We now consider games with limited violations, i.e. (reachability)
games with \LVenergynb, \LVenergyconsnb and \LVenergysum
objectives. We address the problems of deciding the winner
in the one-player and two-player settings, and consider the existence and
minimization questions.


\begin{restatable}{theorem}{thmtwfive}
\label{thm_apna_LUHV_2P}
\LVenergynb, \LVenergyconsnb and \LVenergysum (reachability) games
are \PSPACE-complete for one-player arenas, and \EXPTIME-complete for
two-player arenas.
\end{restatable}

\begin{proof}
    Membership in \PSPACE and \EXPTIME can be obtained by building a
  variant \exGLUV of the \exGLU arena: besides storing the energy level in
  each state, we~can also store the amount of violations (for any of
  the three measures we consider). More precisely, given an arena~$G$,
  lower and upper bounds~$L$ and~$U$ on the energy level, a soft bound~$S$,
  and a
  bound~$V$ on the measure of violations,
  we~define a new arena\footnote{In order to factor our
    proof, we~store all three measures of violations in one single
    arena, even if only one measure per type of \LVenergy game is needed.}~\exGLUV with set of
  states 
  $(Q\times([L;U]\cup\{\bot\})\times ([0;V]\cup\{\bot\})^3)$,   
  and each transition~$(q,w,q')$ of the
  original arena generates a transition from state
  $(q,l,(n,c,s))$ to state~$(q',l',(n',c',s'))$ whenever
  \begin{itemize}
  \item $l'$ correctly encodes the evolution of the energy level: if $l$ and~$l+w$ are in~$[L;U]$, then $l'=l+w$; if the energy level leaves interval $[L;U]$ ($l+w<L$ or $l+w>U$) or has formerly leaved interval $[L;U]$ (in this case $l=\bot$), then $l'=\bot$.
  \item $n'$ correctly stores the number of violations: $n'=\bot$ if $l'\in(S;U]$ and $n+1>V$ (the number of violations allowed is exceeded); once the number of violations is exceeded ($n=\bot$) or the maximal energy level is exceeded ($l'=\bot$), we have $n'=\bot$; last, $n'=n$ if $l'\in [L;S]$ (the current state does not violate bound $S$), and $n'=n+1$ if $l'\in(S;U]$ and $n+1\leq V$ (the current state is an additional violation of soft bound $S$);

\end{itemize}
    
We can similarly update $c'$ to count the current number of
consecutive violations, and $s'$ for the sum of all violations.
In~the resulting arena, values~$n$, $c$ and~$s$ keep tack of the number of
violations, number of consecutive violations and sum of violations;
their values range in $[0,V]$ and are set to~$\bot$ as soon as they
exceed bound $V$, or if the energy level has exceeded its
bounds. The~arena~\exGLUV has size exponential,
and \LVenergyall-reachability problems can then be reduced to solving
reachability of corresponding sets of states in~\exGLUV.  The
announced complexity results follow. Hardness results are obtained by
a straightforward encoding of \LUenergy reachability problems, taking
$S=U$ and $V=0$.


Solving \LVenergynb, \LVenergyconsnb, \LVenergysum games (without
reachability objective) can be performed with arena \exGLUV built
above. Now, the objective
in \LVenergynb, \LVenergyconsnb, \LVenergysum games is to enforce
infinite runs, that avoid states with $l=\bot$ and with $n=\bot$,
$c=\bot$ or $s=\bot$, depending on the chosen measure of
violations. 
\end{proof}




When the strong upper bound~$U$ is not given, the existence problem
consists in deciding if such a bound exists under which \Pl1 wins the
\LSUenergyall game. We~have:
\begin{restatable}{theorem}{thmtwsix}
\label{thm_Apnar_exists_min_exptimec}
The existence problems for \LVenergynb, \LVenergyconsnb, and
\LVenergysum (reachability) games are \PSPACE complete for the
one-player case and \EXPTIME-complete for the two-player case.
\end{restatable}

\begin{proof}
  Along
  any outcome of a~winning strategy, the energy level remains
  below $S+V\cdot w_{\max}$, where $w_{\max}$ is the maximal weight
  appearing on transitions of~$G$. 
  This
gives a strong upper bound~$U$, with which we can apply
  the construction above and check the existence of a winning strategy
  for \Pl1. 
\end{proof}


\begin{restatable}{theorem}{thmtwseven}
\label{thm_minimization}
Let $G$ be an arena, $L$ and $S$ be integer bounds, and $V_{\max}$
be an integer. There exist algorithms that compute the value of $U$ (if~any)
that minimizes the value of~$V$ (below~$V_{\max}$) for which \Pl1 has a winning
strategy in a \LVenergyall
(reachability) game.  These algorithms run in \PSPACE for one-player
games and in \EXPTIME for two-player games. These bounds are sharp.
\end{restatable}

\begin{proof}
We perform a binary search for an optimal value for $V$ when the
strict upper energy bound~$U$ varies between~$S$ and
$S+V_{\max}.w_{\max}$. For~each value~$U$, we discard from \exGLUV 
transitions for which the energy level exceeds~$U$.  One can remark
that when $U$ grows, the minimal amount of violation may decrease, both
in reachability and infinite-run games.  We can hence discover optimal
values with a polynomial number of \PSPACE checks (for the one-player
games), and \EXPTIME checks in the two-player case.
\end{proof}

%
%

%

\vspace{-\medskipamount}
\section{Conclusion}
\label{section_conclusion}

This paper has considered several variants of energy games. The first
variant defines games with upper and lower bound constraints, combined
with reachability objectives.  The second variant
defines games with a strong lower bound and a soft upper bound, which can
be temporarily exceeded.
In~the one player case,
complexities ranges from \PTIME to \PSPACE-complete, and in the
two-player case from $\NP\cap co\NP$ to \EXPTIME-complete.  In
general, the complexity is the same for a reachability and for an
infinite run objective. Interestingly, for \LWenergy games, the
complexity of the single player case is $\PTIME$, but reachability
objectives require exponential memory (in the size of the weak upper
bound) while strategies are memoryless for infinite run objectives.

A possible extension of this work is to consider energy games with
mean-payoff functions and discounted total payoff, both for the energy
level and for the violation constraints, and the associated
minimization and existence problems.

\paragraph{Acknowledgements.}
This work was supported by UMI Relax.
We~thank the anonymous reviewers for their
suggestions, which helped us improve the presentation.



\begin{thebibliography}{10}
\providecommand{\bibitemdeclare}[2]{}
\providecommand{\surnamestart}{}
\providecommand{\surnameend}{}
\providecommand{\urlprefix}{Available at }
\providecommand{\url}[1]{\texttt{#1}}
\providecommand{\href}[2]{\texttt{#2}}
\providecommand{\urlalt}[2]{\href{#1}{#2}}
\providecommand{\realdoi}[1]{doi:\urlalt{http://dx.doi.org/#1}{#1}}\catcode`\_8\relax
\providecommand{\doi}{\catcode`\_11\relax\realdoi}
\providecommand{\bibinfo}[2]{#2}

\bibitemdeclare{inproceedings}{Andersson06}
\bibitem{Andersson06}
\bibinfo{author}{Daniel \surnamestart Andersson\surnameend}
  (\bibinfo{year}{2006}): \emph{\bibinfo{title}{An improved algorithm for
  discounted payoff games}}.
\newblock In \bibinfo{editor}{Janneke \surnamestart Huitink\surnameend} \&
  \bibinfo{editor}{Sophia \surnamestart Katrenko\surnameend}, editors: {\sl
  \bibinfo{booktitle}{{P}roceedings of the 11th {ESSLLI} {S}tudent {S}ession}},
  pp. \bibinfo{pages}{91--98}.

\bibitemdeclare{inproceedings}{BFLM10}
\bibitem{BFLM10}
\bibinfo{author}{Patricia \surnamestart Bouyer\surnameend},
  \bibinfo{author}{Uli \surnamestart Fahrenberg\surnameend},
  \bibinfo{author}{Kim~Guldstrand \surnamestart Larsen\surnameend} \&
  \bibinfo{author}{Nicolas \surnamestart Markey\surnameend}
  (\bibinfo{year}{2010}): \emph{\bibinfo{title}{Timed Automata with Observers
  under Energy Constraints}}.
\newblock In \bibinfo{editor}{Karl~Henrik \surnamestart Johansson\surnameend}
  \& \bibinfo{editor}{Wang \surnamestart Yi\surnameend}, editors: {\sl
  \bibinfo{booktitle}{{P}roceedings of the 13th {I}nternational {W}orkshop on
  {H}ybrid {S}ystems: {C}omputation and {C}ontrol ({HSCC}'10)}},
  \bibinfo{publisher}{ACM Press}, pp. \bibinfo{pages}{61--70},
  \doi{10.1145/1755952.1755963}.

\bibitemdeclare{inproceedings}{BouyerFLMS08}
\bibitem{BouyerFLMS08}
\bibinfo{author}{Patricia \surnamestart Bouyer\surnameend},
  \bibinfo{author}{Uli \surnamestart Fahrenberg\surnameend},
  \bibinfo{author}{Kim~Guldstrand \surnamestart Larsen\surnameend},
  \bibinfo{author}{Nicolas \surnamestart Markey\surnameend} \&
  \bibinfo{author}{Ji{\v r}{\'\i} \surnamestart Srba\surnameend}
  (\bibinfo{year}{2008}): \emph{\bibinfo{title}{Infinite Runs in Weighted Timed
  Automata with Energy Constraints}}.
\newblock In \bibinfo{editor}{Franck \surnamestart Cassez\surnameend} \&
  \bibinfo{editor}{Claude \surnamestart Jard\surnameend}, editors: {\sl
  \bibinfo{booktitle}{{P}roceedings of the 6th {I}nternational {C}onferences on
  {F}ormal {M}odelling and {A}nalysis of {T}imed {S}ystems ({FORMATS}'08)}},
  {\sl \bibinfo{series}{Lecture Notes in Computer Science}}
  \bibinfo{volume}{5215}, \bibinfo{publisher}{Springer-Verlag}, pp.
  \bibinfo{pages}{33--47}, \doi{10.1007/978-3-540-85778-5_4}.

\bibitemdeclare{inproceedings}{BHMRZ17}
\bibitem{BHMRZ17}
\bibinfo{author}{Patricia \surnamestart Bouyer\surnameend},
  \bibinfo{author}{Piotr \surnamestart Hofman\surnameend},
  \bibinfo{author}{Nicolas \surnamestart Markey\surnameend},
  \bibinfo{author}{Mickael \surnamestart Randour\surnameend} \&
  \bibinfo{author}{Martin \surnamestart Zimmermann\surnameend}
  (\bibinfo{year}{2017}): \emph{\bibinfo{title}{Bounding Average-energy
  Games}}.
\newblock In \bibinfo{editor}{Javier \surnamestart Esparza\surnameend} \&
  \bibinfo{editor}{Andrzej \surnamestart Murawski\surnameend}, editors: {\sl
  \bibinfo{booktitle}{{P}roceedings of the 20th {I}nternational {C}onference on
  {F}oundations of {S}oftware {S}cience and {C}omputation {S}tructure
  ({FoSSaCS}'17)}}, {\sl \bibinfo{series}{Lecture Notes in Computer Science}}
  \bibinfo{volume}{10203}, \bibinfo{publisher}{Springer-Verlag}, pp.
  \bibinfo{pages}{179--195}, \doi{10.1007/978-3-662-54458-7_11}.

\bibitemdeclare{inproceedings}{qest2012-BLM}
\bibitem{qest2012-BLM}
\bibinfo{author}{Patricia \surnamestart Bouyer\surnameend},
  \bibinfo{author}{Kim~Guldstrand \surnamestart Larsen\surnameend} \&
  \bibinfo{author}{Nicolas \surnamestart Markey\surnameend}
  (\bibinfo{year}{2012}): \emph{\bibinfo{title}{Lower-Bound Constrained Runs in
  Weighted Timed Automata}}.
\newblock In: {\sl \bibinfo{booktitle}{{P}roceedings of the 9th {I}nternational
  {C}onference on {Q}uantitative {E}valuation of {S}ystems ({QEST}'12)}},
  \bibinfo{publisher}{IEEE Comp. Soc. Press}, pp. \bibinfo{pages}{128--137},
  \doi{10.1109/QEST.2012.28}.

\bibitemdeclare{inproceedings}{BMRLL15}
\bibitem{BMRLL15}
\bibinfo{author}{Patricia \surnamestart Bouyer\surnameend},
  \bibinfo{author}{Nicolas \surnamestart Markey\surnameend},
  \bibinfo{author}{Mickael \surnamestart Randour\surnameend},
  \bibinfo{author}{Kim~Guldstrand \surnamestart Larsen\surnameend} \&
  \bibinfo{author}{Simon \surnamestart Laursen\surnameend}
  (\bibinfo{year}{2015}): \emph{\bibinfo{title}{Average-energy games}}.
\newblock In \bibinfo{editor}{Javier \surnamestart Esparza\surnameend} \&
  \bibinfo{editor}{Enrico \surnamestart Tronci\surnameend}, editors: {\sl
  \bibinfo{booktitle}{{P}roceedings of the 6th {I}nternational {S}ymposium on
  {G}ames, {A}utomata, {L}ogics and {F}ormal {V}erification ({GandALF}'15)}},
  {\sl \bibinfo{series}{Electronic Proceedings in Theoretical Computer
  Science}} \bibinfo{volume}{193}, pp. \bibinfo{pages}{1--15},
  \doi{10.4204/EPTCS.193.1}.

\bibitemdeclare{inproceedings}{CJLRR09}
\bibitem{CJLRR09}
\bibinfo{author}{Franck \surnamestart Cassez\surnameend},
  \bibinfo{author}{Jan~J. \surnamestart Jensen\surnameend},
  \bibinfo{author}{Kim~Guldstrand \surnamestart Larsen\surnameend},
  \bibinfo{author}{Jean-Fran{\c c}ois \surnamestart Raskin\surnameend} \&
  \bibinfo{author}{Pierre-Alain \surnamestart Reynier\surnameend}
  (\bibinfo{year}{2009}): \emph{\bibinfo{title}{Automatic Synthesis of Robust
  and Optimal Controllers~-- An~Industrial Case Study}}.
\newblock In \bibinfo{editor}{Rupak \surnamestart Majumdar\surnameend} \&
  \bibinfo{editor}{Paulo \surnamestart Tabuada\surnameend}, editors: {\sl
  \bibinfo{booktitle}{{P}roceedings of the 12th {I}nternational {W}orkshop on
  {H}ybrid {S}ystems: {C}omputation and {C}ontrol ({HSCC}'09)}}, {\sl
  \bibinfo{series}{Lecture Notes in Computer Science}} \bibinfo{volume}{5469},
  \bibinfo{publisher}{Springer-Verlag}, pp. \bibinfo{pages}{90--104},
  \doi{10.1007/978-3-642-00602-9_7}.

\bibitemdeclare{inproceedings}{CdAHS03}
\bibitem{CdAHS03}
\bibinfo{author}{Arindam \surnamestart Chakrabarti\surnameend},
  \bibinfo{author}{Luca \surnamestart de~Alfaro\surnameend},
  \bibinfo{author}{Thomas~A. \surnamestart Henzinger\surnameend} \&
  \bibinfo{author}{Mari{\"e}lle \surnamestart Stoelinga\surnameend}
  (\bibinfo{year}{2003}): \emph{\bibinfo{title}{Resource Interfaces}}.
\newblock In \bibinfo{editor}{Rajeev \surnamestart Alur\surnameend} \&
  \bibinfo{editor}{Insup \surnamestart Lee\surnameend}, editors: {\sl
  \bibinfo{booktitle}{{P}roceedings of the 3rd {I}nternational {C}onference on
  {E}mbedded {S}oftware ({EMSOFT}'03)}}, {\sl \bibinfo{series}{Lecture Notes in
  Computer Science}} \bibinfo{volume}{2855},
  \bibinfo{publisher}{Springer-Verlag}, pp. \bibinfo{pages}{117--133},
  \doi{10.1007/978-3-540-45212-6_9}.

\bibitemdeclare{article}{ChatterjeeD12}
\bibitem{ChatterjeeD12}
\bibinfo{author}{Krishnendu \surnamestart Chatterjee\surnameend} \&
  \bibinfo{author}{Laurent \surnamestart Doyen\surnameend}
  (\bibinfo{year}{2012}): \emph{\bibinfo{title}{Energy Parity Games}}.
\newblock {\sl \bibinfo{journal}{Theoretical Computer Science}}
  \bibinfo{volume}{458}, pp. \bibinfo{pages}{49--60},
  \doi{10.1016/j.tcs.2012.07.038}.

\bibitemdeclare{inproceedings}{Chatterjee0H17}
\bibitem{Chatterjee0H17}
\bibinfo{author}{Krishnendu \surnamestart Chatterjee\surnameend},
  \bibinfo{author}{Laurent \surnamestart Doyen\surnameend} \&
  \bibinfo{author}{Thomas~A. \surnamestart Henzinger\surnameend}
  (\bibinfo{year}{2017}): \emph{\bibinfo{title}{The Cost of Exactness in
  Quantitative Reachability}}.
\newblock In \bibinfo{editor}{Luca \surnamestart Aceto\surnameend},
  \bibinfo{editor}{Giorgio \surnamestart Bacci\surnameend},
  \bibinfo{editor}{Giovanni \surnamestart Bacci\surnameend},
  \bibinfo{editor}{Anna \surnamestart Ing{\'o}lfsd{\'o}ttir\surnameend},
  \bibinfo{editor}{Axel \surnamestart Legay\surnameend} \&
  \bibinfo{editor}{Radu \surnamestart Mardare\surnameend}, editors: {\sl
  \bibinfo{booktitle}{Models, Algorithms, Logics and Tools: Essays Dedicated to
  Kim Guldstrand Larsen on the Occasion of His 60th Birthday}}, {\sl
  \bibinfo{series}{Lecture Notes in Computer Science}} \bibinfo{volume}{10460},
  \bibinfo{publisher}{Springer-Verlag}, pp. \bibinfo{pages}{367--381},
  \doi{10.1007/978-3-319-63121-9_18}.

\bibitemdeclare{inproceedings}{ChatterjeeDHR10}
\bibitem{ChatterjeeDHR10}
\bibinfo{author}{Krishnendu \surnamestart Chatterjee\surnameend},
  \bibinfo{author}{Laurent \surnamestart Doyen\surnameend},
  \bibinfo{author}{Thomas~A. \surnamestart Henzinger\surnameend} \&
  \bibinfo{author}{Jean-Fran{\c c}ois \surnamestart Raskin\surnameend}
  (\bibinfo{year}{2010}): \emph{\bibinfo{title}{Generalized Mean-payoff and
  Energy Games}}.
\newblock In \bibinfo{editor}{Kamal \surnamestart Lodaya\surnameend} \&
  \bibinfo{editor}{Meena \surnamestart Mahajan\surnameend}, editors: {\sl
  \bibinfo{booktitle}{{P}roceedings of the 30th {C}onference on {F}oundations
  of {S}oftware {T}echnology and {T}heoretical {C}omputer {S}cience
  ({FSTTCS}'10)}}, {\sl \bibinfo{series}{Leibniz International Proceedings in
  Informatics}}~\bibinfo{volume}{8}, \bibinfo{publisher}{Leibniz-Zentrum
  f{\"u}r Informatik}, pp. \bibinfo{pages}{505--516},
  \doi{10.4230/LIPIcs.FSTTCS.2010.505}.

\bibitemdeclare{article}{ChatterjeeRR14}
\bibitem{ChatterjeeRR14}
\bibinfo{author}{Krishnendu \surnamestart Chatterjee\surnameend},
  \bibinfo{author}{Mickael \surnamestart Randour\surnameend} \&
  \bibinfo{author}{Jean-Fran{\c c}ois \surnamestart Raskin\surnameend}
  (\bibinfo{year}{2014}): \emph{\bibinfo{title}{Strategy Synthesis for
  Multi-dimensional Quantitative Objectives}}.
\newblock {\sl \bibinfo{journal}{Acta Informatica}}
  \bibinfo{volume}{51}(\bibinfo{number}{3-4}), pp. \bibinfo{pages}{129--163},
  \doi{10.1007/s00236-013-0182-6}.

\bibitemdeclare{inproceedings}{DDGRT10}
\bibitem{DDGRT10}
\bibinfo{author}{Aldric \surnamestart Degorre\surnameend},
  \bibinfo{author}{Laurent \surnamestart Doyen\surnameend},
  \bibinfo{author}{Raffaella \surnamestart Gentilini\surnameend},
  \bibinfo{author}{Jean-Fran{\c c}ois \surnamestart Raskin\surnameend} \&
  \bibinfo{author}{Szymon \surnamestart Toru{\'n}czyk\surnameend}
  (\bibinfo{year}{2010}): \emph{\bibinfo{title}{Energy and Mean-Payoff Games
  with Imperfect Information}}.
\newblock In \bibinfo{editor}{Anuj \surnamestart Dawar\surnameend} \&
  \bibinfo{editor}{Helmut \surnamestart Veith\surnameend}, editors: {\sl
  \bibinfo{booktitle}{{P}roceedings of the 24th {I}nternational {W}orkshop on
  {C}omputer {S}cience {L}ogic ({CSL}'10)}}, {\sl \bibinfo{series}{Lecture
  Notes in Computer Science}} \bibinfo{volume}{6247},
  \bibinfo{publisher}{Springer-Verlag}, pp. \bibinfo{pages}{260--274},
  \doi{10.1007/978-3-642-15205-4_22}.

\bibitemdeclare{inproceedings}{DM18}
\bibitem{DM18}
\bibinfo{author}{Dario \surnamestart Della{ }Monica\surnameend} \&
  \bibinfo{author}{Aniello \surnamestart Murano\surnameend}
  (\bibinfo{year}{2018}): \emph{\bibinfo{title}{Parity-energy {ATL} for
  Qualitative and Quantitative Reasoning in~{MAS}}}.
\newblock In \bibinfo{editor}{Elisabeth \surnamestart Andr{\'e}\surnameend},
  \bibinfo{editor}{Sven \surnamestart Koenig\surnameend},
  \bibinfo{editor}{Mehdi \surnamestart Dastani\surnameend} \&
  \bibinfo{editor}{Gita \surnamestart Sukthankar\surnameend}, editors: {\sl
  \bibinfo{booktitle}{{P}roceedings of the 17th {I}nternational {C}onference on
  {A}utonomous {A}gents and {M}ultiagent {S}ystems ({AAMAS}'18)}},
  \bibinfo{publisher}{International Foundation for Autonomous Agents and
  Multiagent Systems}, pp. \bibinfo{pages}{1441--1449}.

\bibitemdeclare{article}{EM79}
\bibitem{EM79}
\bibinfo{author}{Andrzej \surnamestart Ehrenfeucht\surnameend} \&
  \bibinfo{author}{Jan \surnamestart Mycielski\surnameend}
  (\bibinfo{year}{1979}): \emph{\bibinfo{title}{Positional strategies for mean
  payoff games}}.
\newblock {\sl \bibinfo{journal}{International Journal of Game Theory}}
  \bibinfo{volume}{8}(\bibinfo{number}{2}), pp. \bibinfo{pages}{109--113},
  \doi{10.1007/BF01768705}.

\bibitemdeclare{mastersthesis}{EjsingDuun2013InfiniteRI}
\bibitem{EjsingDuun2013InfiniteRI}
\bibinfo{author}{Daniel \surnamestart Ejsing{-}Dunn\surnameend} \&
  \bibinfo{author}{Lisa \surnamestart Fontani\surnameend}
  (\bibinfo{year}{2013}): \emph{\bibinfo{title}{Infinite Runs in Recharge
  Automata}}.
\newblock Master's thesis, \bibinfo{school}{Computer Science Department,
  Aalborg University, Denmark}.

\bibitemdeclare{inproceedings}{FahrenbergJLS11}
\bibitem{FahrenbergJLS11}
\bibinfo{author}{Uli \surnamestart Fahrenberg\surnameend},
  \bibinfo{author}{Line \surnamestart Juhl\surnameend},
  \bibinfo{author}{Kim~Guldstrand \surnamestart Larsen\surnameend} \&
  \bibinfo{author}{Ji{\v r}{\'\i} \surnamestart Srba\surnameend}
  (\bibinfo{year}{2011}): \emph{\bibinfo{title}{Energy Games in Multiweighted
  Automata}}.
\newblock In \bibinfo{editor}{Antonio \surnamestart Cerone\surnameend} \&
  \bibinfo{editor}{Pekka \surnamestart Pihlajasaari\surnameend}, editors: {\sl
  \bibinfo{booktitle}{{P}roceedings of the 8th {I}nternational {C}olloquium on
  {T}heoretical {A}spects of {C}omputing ({ICTAC}'11)}}, {\sl
  \bibinfo{series}{Lecture Notes in Computer Science}} \bibinfo{volume}{6916},
  \bibinfo{publisher}{Springer-Verlag}, pp. \bibinfo{pages}{95--115},
  \doi{10.1007/978-3-642-23283-1_9}.

\bibitemdeclare{inproceedings}{FearnleyJ13}
\bibitem{FearnleyJ13}
\bibinfo{author}{John \surnamestart Fearnley\surnameend} \&
  \bibinfo{author}{Marcin \surnamestart Jurdzi{\'n}ski\surnameend}
  (\bibinfo{year}{2013}): \emph{\bibinfo{title}{Reachability in two-clock timed
  automata is {PSPACE}-complete}}.
\newblock In \bibinfo{editor}{Fedor~V. \surnamestart Fomin\surnameend},
  \bibinfo{editor}{Rusins \surnamestart Freivalds\surnameend},
  \bibinfo{editor}{Marta \surnamestart Kwiatkowska\surnameend} \&
  \bibinfo{editor}{David \surnamestart Peleg\surnameend}, editors: {\sl
  \bibinfo{booktitle}{{P}roceedings of the 40th {I}nternational {C}olloquium on
  {A}utomata, {L}anguages and {P}rogramming ({ICALP}'13)~-- Part~{II}}}, {\sl
  \bibinfo{series}{Lecture Notes in Computer Science}} \bibinfo{volume}{7966},
  \bibinfo{publisher}{Springer-Verlag}, pp. \bibinfo{pages}{212--223},
  \doi{10.1007/978-3-642-39212-2_21}.

\bibitemdeclare{inproceedings}{GHOW10}
\bibitem{GHOW10}
\bibinfo{author}{Stefan \surnamestart G{\"o}ller\surnameend},
  \bibinfo{author}{Christoph \surnamestart Haase\surnameend},
  \bibinfo{author}{Jo{\"e}l \surnamestart Ouaknine\surnameend} \&
  \bibinfo{author}{James \surnamestart Worrell\surnameend}
  (\bibinfo{year}{2010}): \emph{\bibinfo{title}{Model Checking Succinct and
  Parametric One-Counter Automata}}.
\newblock In \bibinfo{editor}{Samson \surnamestart Abramsky\surnameend},
  \bibinfo{editor}{Cyril \surnamestart Gavoille\surnameend},
  \bibinfo{editor}{Claude \surnamestart Kirchner\surnameend},
  \bibinfo{editor}{Friedhelm \surnamestart Meyer auf~der Heide\surnameend} \&
  \bibinfo{editor}{Paul~G. \surnamestart Spirakis\surnameend}, editors: {\sl
  \bibinfo{booktitle}{{P}roceedings of the 37th {I}nternational {C}olloquium on
  {A}utomata, {L}anguages and {P}rogramming ({ICALP}'10)~-- Part~{II}}}, {\sl
  \bibinfo{series}{Lecture Notes in Computer Science}} \bibinfo{volume}{6199},
  \bibinfo{publisher}{Springer-Verlag}, pp. \bibinfo{pages}{575--586},
  \doi{10.1007/978-3-642-14162-1_48}.

\bibitemdeclare{inproceedings}{Hun15}
\bibitem{Hun15}
\bibinfo{author}{Paul \surnamestart Hunter\surnameend} (\bibinfo{year}{2015}):
  \emph{\bibinfo{title}{Reachability in Succinct One-Counter Games}}.
\newblock In \bibinfo{editor}{Miko{\l}aj \surnamestart
  Boja{\'n}czyk\surnameend}, \bibinfo{editor}{S{\l}awomir \surnamestart
  Lasota\surnameend} \& \bibinfo{editor}{Igor \surnamestart
  Potapov\surnameend}, editors: {\sl \bibinfo{booktitle}{{P}roceedings of the
  9th {W}orkshop on {R}eachability {P}roblems in {C}omputational {M}odels
  ({RP}'15)}}, {\sl \bibinfo{series}{Lecture Notes in Computer Science}}
  \bibinfo{volume}{9328}, \bibinfo{publisher}{Springer-Verlag}, pp.
  \bibinfo{pages}{37--49}, \doi{10.1007/978-3-319-24537-9_5}.

\bibitemdeclare{inproceedings}{JLR13}
\bibitem{JLR13}
\bibinfo{author}{Line \surnamestart Juhl\surnameend},
  \bibinfo{author}{Kim~Guldstrand \surnamestart Larsen\surnameend} \&
  \bibinfo{author}{Jean-Fran{\c c}ois \surnamestart Raskin\surnameend}
  (\bibinfo{year}{2013}): \emph{\bibinfo{title}{Optimal Bounds for
  Multiweighted and Parametrised Energy Games}}.
\newblock In \bibinfo{editor}{Zhiming \surnamestart Liu\surnameend},
  \bibinfo{editor}{Jim \surnamestart Woodcock\surnameend} \&
  \bibinfo{editor}{Yunshan \surnamestart Zhu\surnameend}, editors: {\sl
  \bibinfo{booktitle}{Theories of Programming and Formal Methods~-- Essays
  Dedicated to Jifeng He on the Occasion of His 70th Birthday}}, {\sl
  \bibinfo{series}{Lecture Notes in Computer Science}} \bibinfo{volume}{8051},
  \bibinfo{publisher}{Springer-Verlag}, pp. \bibinfo{pages}{244--255},
  \doi{10.1007/978-3-642-39698-4_15}.

\bibitemdeclare{inproceedings}{JurdzinskiLS07}
\bibitem{JurdzinskiLS07}
\bibinfo{author}{Marcin \surnamestart Jurdzi{\'n}ski\surnameend},
  \bibinfo{author}{Fran{\c c}ois \surnamestart Laroussinie\surnameend} \&
  \bibinfo{author}{Jeremy \surnamestart Sproston\surnameend}
  (\bibinfo{year}{2007}): \emph{\bibinfo{title}{Model Checking Probabilistic
  Timed Automata with One or Two Clocks}}.
\newblock In \bibinfo{editor}{Orna \surnamestart Grumberg\surnameend} \&
  \bibinfo{editor}{Michael \surnamestart Huth\surnameend}, editors: {\sl
  \bibinfo{booktitle}{{P}roceedings of the 13th {I}nternational {C}onference on
  {T}ools and {A}lgorithms for {C}onstruction and {A}nalysis of {S}ystems
  ({TACAS}'07)}}, {\sl \bibinfo{series}{Lecture Notes in Computer Science}}
  \bibinfo{volume}{4424}, \bibinfo{publisher}{Springer-Verlag}, pp.
  \bibinfo{pages}{170--184}, \doi{10.1007/978-3-540-71209-1_15}.

\bibitemdeclare{inproceedings}{JLS15}
\bibitem{JLS15}
\bibinfo{author}{Marcin \surnamestart Jurdzi{\'n}ski\surnameend},
  \bibinfo{author}{Ranko \surnamestart Lazi{\'c}\surnameend} \&
  \bibinfo{author}{Sylvain \surnamestart Schmitz\surnameend}
  (\bibinfo{year}{2015}): \emph{\bibinfo{title}{Fixed-Dimensional Energy Games
  are in Pseudo-Polynomial Time}}.
\newblock In \bibinfo{editor}{Magn{\'u}s~M. \surnamestart
  Halld{\'o}rsson\surnameend}, \bibinfo{editor}{Kazuo \surnamestart
  Iwana\surnameend}, \bibinfo{editor}{Naoki \surnamestart Kobayashi\surnameend}
  \& \bibinfo{editor}{Bettina \surnamestart Speckmann\surnameend}, editors:
  {\sl \bibinfo{booktitle}{{P}roceedings of the 42nd {I}nternational
  {C}olloquium on {A}utomata, {L}anguages and {P}rogramming ({ICALP}'15)~--
  Part~{II}}}, {\sl \bibinfo{series}{Lecture Notes in Computer Science}}
  \bibinfo{volume}{9135}, \bibinfo{publisher}{Springer-Verlag}, pp.
  \bibinfo{pages}{260--272}, \doi{10.1007/978-3-662-47666-6_21}.

\bibitemdeclare{article}{Mar75}
\bibitem{Mar75}
\bibinfo{author}{Donald~A. \surnamestart Martin\surnameend}
  (\bibinfo{year}{1975}): \emph{\bibinfo{title}{{B}orel Determinacy}}.
\newblock {\sl \bibinfo{journal}{Annals of Mathematics}}
  \bibinfo{volume}{102}(\bibinfo{number}{2}), pp. \bibinfo{pages}{363--371},
  \doi{10.2307/1971035}.

\bibitemdeclare{article}{Reichert16}
\bibitem{Reichert16}
\bibinfo{author}{Julien \surnamestart Reichert\surnameend}
  (\bibinfo{year}{2016}): \emph{\bibinfo{title}{On The Complexity of Counter
  Reachability Games}}.
\newblock {\sl \bibinfo{journal}{Fundamenta Informaticae}}
  \bibinfo{volume}{143}(\bibinfo{number}{3-4}), pp. \bibinfo{pages}{415--436},
  \doi{10.3233/FI-2016-1320}.

\bibitemdeclare{article}{VCDHRR15}
\bibitem{VCDHRR15}
\bibinfo{author}{Yaron \surnamestart Velner\surnameend},
  \bibinfo{author}{Krishnendu \surnamestart Chatterjee\surnameend},
  \bibinfo{author}{Laurent \surnamestart Doyen\surnameend},
  \bibinfo{author}{Thomas~A. \surnamestart Henzinger\surnameend},
  \bibinfo{author}{Alexander \surnamestart Rabinovich\surnameend} \&
  \bibinfo{author}{Jean-Fran{\c c}ois \surnamestart Raskin\surnameend}
  (\bibinfo{year}{2015}): \emph{\bibinfo{title}{The complexity of
  multi-mean-payoff and multi-energy games}}.
\newblock {\sl \bibinfo{journal}{Information and Computation}}
  \bibinfo{volume}{241}, pp. \bibinfo{pages}{177--196},
  \doi{10.1016/j.ic.2015.03.001}.

\bibitemdeclare{article}{ZwickP95}
\bibitem{ZwickP95}
\bibinfo{author}{Uri \surnamestart Zwick\surnameend} \& \bibinfo{author}{Mike
  \surnamestart Paterson\surnameend} (\bibinfo{year}{1996}):
  \emph{\bibinfo{title}{The Complexity of Mean Payoff Games on Graphs}}.
\newblock {\sl \bibinfo{journal}{Theoretical Computer Science}}
  \bibinfo{volume}{158}(\bibinfo{number}{1-2}), pp. \bibinfo{pages}{343--359},
  \doi{10.1016/0304-3975(95)00188-3}.

\end{thebibliography}


\end{document}